\documentclass[journal, onecolumn, draftclsnofoot]{IEEEtran}

\usepackage{verbatim}
\usepackage{amsfonts,amsmath,mathrsfs,amssymb,amsbsy,amsthm}
\usepackage[final]{graphicx}
\usepackage{times}
\usepackage{epstopdf}
\usepackage{subcaption}
\usepackage[font={scriptsize}]{caption}
\usepackage[numbers,sort&compress]{natbib}

\usepackage{lipsum}

\newtheorem{thm}{Theorem}

\newtheorem{lem}{Lemma}

\newtheorem{defn}{Definition}
\newtheorem{remk}{\textit{Remark}}
\newtheorem{claim}{Claim}

\theoremstyle{definition}
\newtheorem{ex}{\textbf{Example}}

\begin{document}
\include{commands}
\title{When Does an Ensemble of Matrices with Randomly Scaled Rows Lose Rank?}

\author{ 
  \IEEEauthorblockN{Navid Naderializadeh, Aly El Gamal, and A. Salman Avestimehr\\}
  \IEEEauthorblockA{
    Department of Electrical Engineering,
    University of Southern California\\
    Emails: naderial@usc.edu, aelgamal@usc.edu, avestimehr@ee.usc.edu} 
}

\maketitle
\begin{abstract}
We consider the problem of determining rank loss conditions for a concatenation of full-rank matrices, such that each row of the composing matrices is scaled by a random coefficient. This problem has applications in wireless interference management and recommendation systems. We determine necessary and sufficient conditions for the design of each matrix, such that the random ensemble will almost surely lose rank by a certain amount. The result is proved by converting the problem to determining rank loss conditions for the union of some specific matroids, and then using tools from matroid and graph theories to derive the necessary and sufficient conditions. As an application, we discuss how this result can be applied to the problem of topological interference management, and characterize the linear symmetric degrees of freedom for a class of network topologies.\end{abstract}

\section{Introduction}\label{sec:intro}
%
We consider the problem in Figure~\ref{fig:mac}, in which $K$ users design matrices $B_1,\ldots,B_K$, where $B_i$ is an $n \times m_i$ matrix, designed by the $i^{\textrm{th}}$ user, with full column rank $m_i \leq n$. A destination obtains the matrix $B_D=\left[\Lambda_1 B_1 \quad \Lambda_2 B_2  \quad \cdots \quad  \Lambda_K B_K \right]$, where $\Lambda_i$ is an $n \times n$ diagonal matrix with random diagonal entries, and the set of all random coefficients is drawn from a continuous joint distribution.
It is easy to see that $R=\min\left(\sum_{i=1}^K m_i, n\right)$ is the maximum rank that the matrix $B_D$ can have. However, by a careful design, one can reduce the rank of this matrix. In particular, the question that we address in this work is: Under what conditions on the design of the matrices $B_1,\ldots,B_K$ will the matrix $B_D$ lose rank by $\tau$ almost surely, i.e., $\text{rank}(B_D) \overset{a.s.}{\leq} R - \tau$, while each of the matrices $B_1,\ldots,B_K$ has full column rank?

\begin{figure}[h]
\centering
\includegraphics[trim = 2.71in 2.95in 1.82in 2.81in, clip,width=0.4\textwidth]{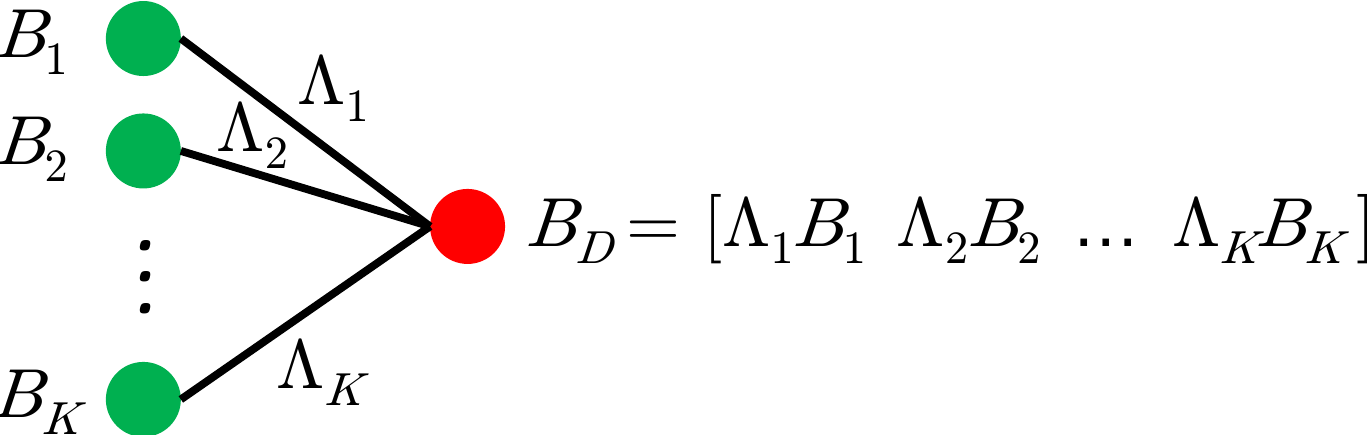}
\caption{Multiple users, each one having a matrix $B_i$ and another user receiving a concatenation of these matrices with randomly-scaled rows.}
\label{fig:mac}
\end{figure}

%

The aforementioned problem arises naturally when understanding the fundamentals of interference management in wireless networks, in particular for the case where no information about the channel state is assumed to be available at the transmitters except for the knowledge about the network topology. In this setting, the $K$ users in Figure \ref{fig:mac} represent transmitters that are interfering at a single receiver. The matrix $B_i$ represents the beamforming matrix at the $i^{\textrm{th}}$ transmitter and the $j^{\textrm{th}}$ diagonal element of the matrix $\Lambda_i$ represents the channel coefficient in the $j^{\textrm{th}}$ time slot. The column span of the matrix $B_D$ represents the space of the received interference. The aforementioned problem is that of determining the conditions on the beamforming matrices that result in an alignment of the interference caused by users $1$ to $K$ in a subspace whose dimension is at most $R-\tau$.

%
Another motivating example arises in recommendation systems, where for a fixed set of items, each matrix $B_i$ represents the ratings given by the $i^{\textrm{th}}$ group of users. 
Each row represents one user's ratings for the set of items, and each diagonal element of $\Lambda_i$ represents a random scaling factor that reflects the user's own bias. For example, one user can give a rating of $10$ to his most favorite movie and $5$ for his least favorite while another user can have a highest rating of $8$ and a lowest rating of $4$. The problem in this setting would be to understand conditions on the ratings of each user group that result in a rank loss of the ensemble. Understanding conditions of rank loss in this application is useful for completion algorithms that recover missing ratings. In other words, the answer of our question in this setting specifies the structure of matrix entries which suffice to complete the entire rating matrix.


As a special case of this problem, one can consider each $B_i$ being only a column vector. This is equivalent to the case where instead of each row, each individual element is scaled by a random coefficient. In this case, the problem can be shown to have a combinatorial structure, depending on the position of zero/non-zero elements of $B_i$'s. In particular, it can be verified that in this case, the ensemble matrix $B_D$ will lose rank by $\tau$ (i.e., $\text{rank}(B_D)\leq R-\tau$) if and only if there does not exist a matching of size greater than $R-\tau$ between rows and columns of $B_D$, where a row is connected to a column if and only if the element with the corresponding row and column indices is non-zero (see e.g.,\cite{path-matching,lovasz,factor} and also the Zippel-Schwartz Lemma \cite{sch,zip}). However, once $B_i$'s are not column vectors, the problem becomes much more complicated since there is a structure in the random scaling; elements of the same row in each matrix are scaled by the same coefficient.


%
In this paper, we characterize a necessary and sufficient condition for rank-loss of the matrix $B_D$. In particular, we determine under what condition on $B_i$'s, the matrix $B_D$ loses rank by $\tau$ almost surely. The proof consists of a succession of equivalence steps. In the first step, we connect the rank of $B_D$ to the rank of the union of certain matroids. This is first done by expanding the determinants of specific submatrices of $B_D$ and then using the Zippel-Schwartz Lemma to express the almost-sure rank loss condition of $B_D$ in terms of the products of the determinants of certain submatrices of $B_i$'s being equal to zero. Then we use the notion of \emph{sparse subspaces} in order to connect the rank loss of $B_D$ to the column span of $B_i$'s and then to the rank loss of the union of carefully-defined matroids. After representing the rank condition of $B_D$ in terms of the rank loss of the union of certain matroids, in the second step, we make use of the matroid union theorem. This theorem enables us to derive an equivalent condition on the structure of $B_1,...,B_K$ for the rank loss of the union of matroids.

In the final step of the proof, we simplify the above condition by constructing appropriate bipartite graphs in which one partite set represents the column vectors of a carefully chosen set of bases for the column span of each of the matrices $B_1,...,B_K$ and the other partite set represents the row indices. A column vector vertex is connected to a row vertex if and only if the vector has a non-zero entry in the corresponding row. We use a variation of the Hall's marriage theorem to reach a final condition expressed through matching sizes on the constructed graphs.

As an application of the derived rank-loss condition, we utilize the result in the context of the topological interference management problem. This problem focuses on the scenario in which interference management primarily relies on a coarse knowledge about channel states in the network, namely the  ``topology'' of the network. Network topology simply refers to the 1-bit feedback information for each link between each transmitter and each receiver, indicating whether or not the signal of the transmitter is received above the noise floor at the corresponding receiver.

There have been several prior works in the literature that have considered similar scenarios. In~\cite{localview}, it is assumed that transmitters are aware of the network topology as well as the actual channel gains within a local neighborhood; the concept of normalized capacity was introduced by characterizing the maximum achievable rate using this channel knowledge as a fraction of the achievable rate using global channel state information. In \cite{jafar}, the authors considered a more restrictive scenario, called topological interference management (TIM), in which the transmitters are only aware of the topology with no information about the channel coefficients. It has been shown  that if the channel gains in the network remain constant for a sufficiently large time, then topological interference management is closely connected to the classical index coding problem (see, e.g., \cite{birk,baryossef,rouayheb,local_coloring_index_coding}), and via this connection, a class of linear interference management schemes has been introduced which rely on the network topology knowledge to align the interference over time at unintended receivers. In~\cite{topology,topology_isit}, the authors considered a class of retransmission-based schemes that only allows transmitters to resend their symbols in order to assist with the neutralization of interference at the receivers. Besides their simplicity, these schemes are robust to channel variations over time and were shown to be optimal in terms of the symmetric degrees-of-freedom (DoF) in many classes of topologies via the outer bounds developed in \cite{topology,topology_isit}.

In this paper, we use our derived rank loss condition in order to solve two problems in the topological interference management framework. In this setting, Figure~\ref{fig:mac} can be seen as a set of $K$ transmitters causing interference at a destination, and the considered rank-loss problem can be seen as reducing the dimension of the subspace occupied by interference at the receiver through careful design of the beamforming matrices at the transmitters.

We first characterize the ``best'' topologies for topological interference management; i.e., the topologies for which half symmetric DoF is achievable for the case of time-varying channels. It is easy to see that in topologies with at least one interference link, the symmetric DoF is upper bounded by $\frac{1}{2}$. Thus, the topologies in which half symmetric DoF is achievable represent the best topologies that one can hope for (from the degrees-of-freedom perspective). For the case where the channel gains in the network are assumed to remain constant for a long-enough period of time, the necessary and sufficient condition on the network topology for achieving half symmetric DoF was characterized in \cite{jafar}. For the case of case of time-varying channels, a sufficient condition for the achievability of half symmetric DoF was derived in \cite{multialign}. In this work, we close the gap in the results on half symmetric DoF by introducing a necessary and sufficient condition under which the symmetric DoF of $\frac{1}{2}$ is achievable for the case of time-varying channels (i.e., without requiring the channels to remain fixed for a long enough time).

Second, we use our condition to characterize the linear symmetric degrees of freedom (DoF) for a class of network topologies with exclusive interference sets. This helps to resolve the characterization of the symmetric DoF for a set of previously open problems considered in~\cite{topology}.

\section{Problem Formulation and Main Result}\label{sec:model}

Consider $K$ matrices $B_1,\ldots,B_K$, where each matrix $B_i, i\in[K]$ has size $n\times m_i$, $m_i\leq n$ (we use $[K]$ to denote the set $\{1,\ldots,K\}$ for any positive integer $K$). Without loss of generality, we assume that each matrix is full-column rank, since the linearly-dependent columns can be removed from each matrix. Furthermore, consider $K$ diagonal matrices $\Lambda_1,\ldots,\Lambda_K$, each of size $n\times n$, where their diagonal elements are drawn from a joint continuous distribution.

In this paper, the problem under consideration is the rank loss of the matrix $B_D=\left[\Lambda_1 B_1 \quad \Lambda_2 B_2  \quad \cdots \quad  \Lambda_K B_K \right]$. To be precise, we aim to find an equivalent condition for when
\begin{align}\label{eq:rankloss}
\text{rank} ([\Lambda_1 B_1 ~~\Lambda_2 B_2~~ \ldots ~~ \Lambda_K B_K])\overset{a.s.}{\leq} R -\tau,
\end{align}
where $R=\min\left(\sum_{i=1}^K m_i,n\right)$ denotes the maximum possible rank of $[\Lambda_1 B_1 ~~\Lambda_2 B_2~~ \ldots ~~ \Lambda_K B_K]$ and $\tau\in \mathbb{Z}^+$.

\textbf{Notation:} For a matrix $B\in\mathbb{R}^{n\times m}$, we use  calligraphic $\mathcal{B}$ to denote the subspace in $\mathbb{R}^n$ spanned by the columns of $B$. Also, for any $X\subseteq[n]$ and $Y\subseteq[m]$, $B_{X,Y}$ denotes the submatrix of $B$ created by removing the rows with indices outside $X$ and removing the columns with indices outside $Y$, and $B_{*,Y}$ denotes the submatrix of $B$ created by removing the columns with indices outside $Y$. Besides, for any $X\subseteq[n]$, $P_X^K$ denotes the set of all partitions of $X$ to $K$ disjoint subsets (each one possibly empty). Finally, for any $J\subseteq [n]$, $J^c$ denotes the complement of $J$ in $[n]$ (i.e., $[n]\setminus J$) and $\mathcal{S}_J$ denotes the subspace of $\mathbb{R}^n$ spanned by the columns of the $n\times n$ identity matrix with indices in $J$. In other words, $\mathcal{S}_J$ is the subspace of $\mathbb{R}^n$ which includes all the vectors that have zero entries in $J^c$. We call $\mathcal{S}_J$ the \emph{sparse subspace} of the set $J$.

We now state the main result in the following theorem.
\begin{thm}\label{thm:main}
The following two statements are equivalent.
\begin{gather}
\emph{rank} ([\Lambda_1 B_1 ~~\Lambda_2 B_2~~ ... ~~ \Lambda_K B_K])\overset{a.s.}{\leq}R -\tau.\tag{C1}\label{eq:original}\\
\forall Y_i\subseteq[m_i], i \in [K] \emph{ s.t. }\sum_{i=1}^K |Y_i|=R, \exists J \subseteq [n]: \sum_{i=1}^K \dim(\mathcal{S}_{J} \cap \mathcal{B}_{i,*,Y_i}) \geq |J|+\tau.\tag{C2}\label{eq:result}
\end{gather}
\end{thm}
\begin{ex}\label{ex:result_clr_1}
As a simple example, consider the case where $K=2$ and $m_1=m_2=\frac{n}{2}$. In this case, $B_1$ and $B_2$ will each be of size $n\times \frac{n}{2}$ and the only possible choice for $Y_1$ and $Y_2$ will be $[\frac{n}{2}]$. Theorem \ref{thm:main} implies that $\text{rank} ([\Lambda_1 B_1 ~~\Lambda_2 B_2])\overset{a.s.}{\leq}n -\tau$ if and only if there exists a set $J\subseteq[n]$ such that,
\begin{align}\label{eq:example}
\dim(\mathcal{S}_{J} \cap \mathcal{B}_{1}) + \dim(\mathcal{S}_{J} \cap \mathcal{B}_{2}) \geq |J|+\tau.
\end{align}
\end{ex}

\begin{ex}
Continuing Example \ref{ex:result_clr_1}, assume $n=2m_1=2m_2=4$ and suppose we fix $B_1=\begin{bmatrix}
1 & 1 & 1 & 0\\
1 & 2 & 3 & 0
\end{bmatrix}^T$.
In this case, Theorem \ref{thm:main} implies that the only way for the matrix $[\Lambda_1 B_1 ~~\Lambda_2 B_2]$ to lose rank by $\tau=1$ almost-surely is that both columns of $B_2$ have zeros in their $4^{th}$ entries. The set $J\subseteq [4]$ satisfying (\ref{eq:example}) would be $\{1,2,3\}$ in this case.
\end{ex}

\begin{remk}
The significance of the above result lies in finding a condition on the structure of the matrices $B_i, i\in[K]$, such that the statistical rank loss condition is met. Checking the structural condition of~\eqref{eq:result} does not involve statistical analysis and relies only on the combinatorial structure of $\mathcal{B}_i$'s.
\end{remk}

\begin{remk}
When each element of the matrices $B_i,\ldots,B_K$ is scaled by a random coefficient,~\eqref{eq:original} corresponds to the size of a bipartite graph matching representing rows and columns of $B_D$ as the two partite sets. This condition will be similar to~\eqref{eq:result}, but instead of $\dim(\mathcal{S}_J \cap \mathcal{B}_i)$, the number of columns in $B_i$ inside the subspace $\mathcal{S}_J$ is considered. This is intuitive as in the considered setting, each row is scaled by the same coefficient, and hence, the angles between column vectors are preserved.
Our result can be seen as a generalization of the classical rank-loss results for random matrices.
\end{remk}

\begin{remk}
We show in Section~\ref{sec:tim} how the result of Theorem~\ref{thm:main} can be used to characterize the linear degrees of freedom for the topological interference management problem of a class of networks that had been considered as an open problem before.
\end{remk}
\section{Proof of the Main Result}\label{sec:result}

The proof of Theorem~\ref{thm:main} is composed of four steps, in each of which we present a condition equivalent to condition (\ref{eq:original}).





\subsection{Step 1: Expansion of Determinant}\label{sec:proofone}

We begin by the following equivalence lemma.

\begin{lem}\label{lem:expansion}
Condition (\ref{eq:original}) is equivalent to the following.
\begin{gather}
\forall X\subseteq[n]: |X|>R-\tau,
\forall Y_i\subseteq[m_i], i\in[K] \text{ s.t. }\sum_{i=1}^K |Y_i|=|X|, \nonumber\\
\forall (I_1,I_2,...,I_K)\in P_X^K \text{ s.t. }|I_i|=|Y_i|:
\prod_{i=1}^K \det(B_{i,I_i,Y_i})=0.\tag{C3}\label{eq:condone}
\end{gather}
\end{lem}

\begin{proof}
By the definition of matrix rank, the condition in~\eqref{eq:original} is equivalent to the fact that any square submatrix of $B_D$ with size greater than $(R-\tau) \times (R-\tau)$ should have a zero determinant almost-surely. This means that for any subset of rows $X\subseteq[n]: |X| > R-\tau$ and any subsets of columns $Y_i\subseteq[m_i]$, $i\in[K] \text{ s.t. }\sum_{i=1}^K |Y_i|=|X|$,
\begin{align*}
&\det\left([\Lambda_{1,X,X} B_{1,X,Y_1} ~~\ldots~~\Lambda_{K,X,X} B_{K,X,Y_K}]\right)\overset{a.s.}{=}0.
\end{align*}

It can be shown that this determinant is composed of monomials in the channel gains whose coefficients are in the form of $\prod_{i=1}^K \det(B_{i,I_i,Y_i})$ for some $(I_1,I_2,...,I_K)\in P_X^K \text{ s.t. }|I_i|=|Y_i|$. By the Zippel-Schwartz Lemma \cite{zip,sch}, for the whole multivariate polynomial to be equal to zero for almost all values of the channel gains, each of these coefficients should be equal to zero, which gives~\eqref{eq:condone}.
\end{proof}

\begin{ex}
Consider the case where $K=2$, $n=3$, $m_1=1$, $m_2=2$, and we have
$B_1=\begin{bmatrix}
a_{11} & a_{21} & a_{31}
\end{bmatrix}^T,
B_2=\begin{bmatrix}
a'_{11} & a'_{21} & a'_{31}\\a'_{12} & a'_{22} & a'_{32}
\end{bmatrix}^T.$
Also, let $\lambda_i$, $\lambda'_i$, $i \in \{1,2,3\}$, be the $i^{\textrm{th}}$ diagonal element of $\Lambda_1$ and $\Lambda_2$, respectively.
We then have,
\begin{align*}
B_D=\begin{bmatrix}
\lambda_1 a_{11} & \lambda'_1 a'_{11} & \lambda'_1 a'_{12} \\ \lambda_2 a_{21} & \lambda'_2 a'_{21} & \lambda'_2 a'_{22} \\ \lambda_3 a_{31} & \lambda'_3 a'_{31} & \lambda'_3 a'_{32}
\end{bmatrix}.
\end{align*}
Now, $\text{rank}(B_D)\leq 2$ is equivalent to $\det(B_D)=0$. Note that,
\begin{align*}
\det(B_D)&=\lambda_1 \lambda'_2 \lambda'_3 \left[a_{11} \det\left(\begin{bmatrix}
a'_{21} & a'_{22}\\a'_{31} & a'_{32}
\end{bmatrix}\right)\right]\\
&-\lambda_2 \lambda'_1 \lambda'_3 \left[a_{21} \det\left(\begin{bmatrix}
a'_{11} & a'_{12}\\a'_{31} & a'_{32}
\end{bmatrix}\right)\right]\\
&+\lambda_3 \lambda'_1 \lambda'_2 \left[a_{31} \det\left(\begin{bmatrix}
a'_{11} & a'_{12}\\a'_{21} & a'_{22}
\end{bmatrix}\right)\right].
\end{align*}
Therefore, the Zippel-Schwartz Lemma implies that this determinant is almost-surely equal to zero if and only if each of the products inside the brackets is equal to zero.
\end{ex}
\subsection{Step 2: Sparse Subspaces}\label{sec:prooftwo}
We now express the condition~\eqref{eq:condone} in terms of sparse subspaces. Recall that a sparse subspace is defined by column vectors that have zero entries for a specific set of rows.

\begin{lem}\label{lem:sparse}
Condition (\ref{eq:condone}) is equivalent to the following.
\begin{gather}
\forall X\subseteq[n]: |X|>R-\tau,\nonumber\\
\forall Y_i\subseteq[m_i], i\in[K] \text{ s.t. }\sum_{i=1}^K |Y_i|=|X|, \nonumber\\
\forall (I_1,I_2,...,I_K)\in P_X^K \text{ s.t. }|I_i|=|Y_i|, \forall i\in[K]:\nonumber\\
\sum_{i=1}^K \dim(\mathcal{S}_{I_i^c} \cap \mathcal{B}_{i,*,Y_i})>0.\tag{C4}\label{eq:condtwo}
\end{gather}
\end{lem}

\begin{proof}
First assume that the condition in~\eqref{eq:condone} holds. Then for some $i\in[K]$ there exists a linear combination of the columns in $B_{i,I_i,Y_i}$ which is equal to the zero vector. Therefore, if we apply this linear combination to the entire matrix $B_{i,*,Y_i}$, we end up with a vector in the sparse subspace $\mathcal{S}_{{I_i}^c}$ since $B_{i,*,Y_i}$ is full-column rank. This implies~\eqref{eq:condtwo}.

If the condition in~\eqref{eq:condtwo} holds, then for some $i\in[K]$, there exists a linear combination of the columns in $B_{i,*,Y_i}$ which is zero in the coordinates in $I_i$. This means that for any $I_i: |I_i|=|Y_i|$, $\det (B_{i,I_i,Y_i})=0$, implying that~\eqref{eq:condone} holds.
\end{proof}

\subsection{Step 3: Connection to Rank of Union of Matroids}\label{sec:proofthree}
In this step, we show how to represent condition (\ref{eq:condtwo}) in terms of the rank loss of the union of certain matroids. To this end, for any choice of $X\subseteq[n]$ and $Y_i\subset[m_i],~\forall i \in [K]$ which satisfy $\sum_{i=1}^K |Y_i|=|X|>R-\tau$ and
\begin{align}\label{eq:matroiddefcondition}
\dim(\mathcal{S}_{X^c} \cap \mathcal{B}_{i,*,Y_i} )=0,\forall i\in[K],
\end{align}

define $\mathcal{I}_{i,X,Y_i}$ as
\begin{align}
\mathcal{I}_{i,X,Y_i}=\{I\subseteq X:\dim(\mathcal{S}_{I\cup X^c}\cap \mathcal{B}_{i,*,Y_i})=0\}, i\in[K].
\end{align}
We have the following claim that is proved in Appendix~\ref{sec:matroid_proof}. 
\begin{claim}\label{matroid}
$M_{i,X,Y_i}=(X,\mathcal{I}_{i,X,Y_i})$ is a matroid with rank function $r_{i,X,Y_i}(J)=|J|-\dim (\mathcal{S}_{J\cup X^c}\cap \mathcal{B}_{i,*,Y_i}) $, $i\in[K]$.
\end{claim}

For any matroid $M_{i,X,Y_i}$, the dual matroid $M_{i,X,Y_i}^*$ is defined as $M_{i,X,Y_i}^*=(X,\mathcal{I}_{i,X,Y_i}^*)$ whose basis sets are the complements of the basis sets of $M_{i,X,Y_i}$ \cite{schrijver}.\footnote{For a matroid $M=(X,\mathcal{I})$, $J\subseteq X$ is called a \emph{basis set} if $J\in \mathcal{I}$ and there is no $J'\in\mathcal{I}$ such that $J\subset J' \subseteq X$ \cite[Ch. 39]{schrijver}.} Using Theorem 39.2 in \cite{schrijver} and Claim \ref{matroid}, the rank of this dual matroid is 
\begin{align}
r_{i,X,Y_i}^*(J)&=|J|-r_{i,X,Y_i}(X)+r_{i,X,Y_i}(X\setminus J)\nonumber\\
&=|J|-(|X|-|Y_i|)+(|X\setminus J|-\dim (\mathcal{S}_{J^c}\cap \mathcal{B}_{i,*,Y_i}))\nonumber\\
&=|Y_i|-\dim (\mathcal{S}_{J^c}\cap \mathcal{B}_{i,*,Y_i}).\label{dualrank}
\end{align}

Now, consider the union of the dual matroids $M_{i,X,Y_i}^*$, $i\in[K]$, denoted by $\bigvee_{i=1}^K M_{i,X,Y_i}^*=\left(X,\bigvee_{i=1}^K \mathcal{I}_{i,X,Y_i}^*\right)$, where

\begin{align*}
\bigvee_{i=1}^K \mathcal{I}_{i,X,Y_i}^*=\left\{\bigcup_{i=1}^K I_i: I_i\in\mathcal{I}_{i,X,Y_i}^*\right\}.
\end{align*}
Let $r_{X,Y_1,...,Y_K}^*(.)$ denote the rank of $\bigvee_{i=1}^K M_{i,X,Y_i}^*$. Then we have the following lemma.
\begin{lem}\label{lem:matroidunion}
For any $X\subseteq [n], Y_i\subseteq[m_i], i\in[K]$ s.t. $\sum_{i=1}^K |Y_i|=|X|> R -\tau$ and (\ref{eq:matroiddefcondition}) is satisfied, the following are equivalent
\begin{gather}
\forall (I_1,I_2,...,I_K)\in P_X^K \text{ s.t. }|I_i|=|Y_i|:
\sum_{i=1}^K \dim(\mathcal{S}_{I_i^c} \cap \mathcal{B}_{i,*,Y_i})>0.\label{eq:c4prime}
\end{gather}
\begin{gather}
r_{X,Y_1,...,Y_K}^*(X)<|X|.\label{eq:condunionmatroid}
\end{gather}
\end{lem}
\begin{proof}
Condition~\eqref{eq:c4prime} is equivalent to the following for any valid choice of $X$ and $Y_i$'s:
\begin{gather}
\nexists (I_1,I_2,...,I_K)\in P_X^K:|I_i|=|Y_i|\text{ and } X\setminus I_i \in \mathcal{I}_{i,X,Y_i}, \forall i\in[K].\label{equiv}
\end{gather}

Now, let us focus on a specific $i\in[K]$ and the corresponding set $I_i$ whose size satisfies $|I_i|=|Y_i|$,
\begin{align}
|X\setminus I_i|&=|X|-|Y_i|\nonumber\\
&=|X|-\dim(\mathcal{S}_{X\cup X^c}\cap \mathcal{B}_{i,*,Y_i})\nonumber\\
&=r_{i,X,Y_i}(X),\label{eq:rank}
\end{align}
where \eqref{eq:rank} follows from Claim \ref{matroid}. By the definition of the rank of a matroid (see, e.g. \cite[Ch. 39]{schrijver}), \eqref{eq:rank} implies that all the members of $\mathcal{I}_{i,X,Y_i}$ (i.e., the independent sets of $M_{i,X,Y_i}$) are of size at most $|X\setminus I_i|$. This means that if $X\setminus I_i\in \mathcal{I}_{i,X,Y_i}$, then it is a basis for $M_{i,X,Y_i}$. This, in turn, is equivalent to $I_i$ being a basis for the dual matroid $M_{i,X,Y_i}^*$.

On the other hand, from (\ref{dualrank}) we have that the rank of the dual matroid $M_{i,X,Y_i}^*$ is upper bounded by $|Y_i|$, implying that all the members of $\mathcal{I}_{i,X,Y_i}^*$ have size at most $|Y_i|$. Thus, under the constraint $|I_i|=|Y_i|$, $I_i$ being a basis for the dual matroid $M_{i,X,Y_i}^*$ is equivalent to $I_i\in \mathcal{I}_{i,X,Y_i}^*$. Consequently, \eqref{equiv} (hence ~\eqref{eq:c4prime}) is equivalent to the following for any valid choice of $X$ and $Y_i$'s:
\begin{gather}
\nexists (I_1,I_2,...,I_K)\in P_X^K:|I_i|=|Y_i|\text{ and }  I_i \in \mathcal{I}_{i,X,Y_i}^*, \forall i\in[K].\label{equiv2}
\end{gather}

Finally, it is easy to verify that (\ref{equiv2}) is equivalent to $r_{X,Y_1,...,Y_K}^*(X)<|X|$, since there cannot exist any $(I_1,I_2,...,I_K)\in P_X^K$ such that $I_i\in \mathcal{I}_{i,X,Y_i}^*$ at the same time for all $i\in[K]$. This completes the proof.
\end{proof}
Note that using this lemma, we can now replace the last two lines of condition~\eqref{eq:condtwo} by simply $r_{X,Y_1,...,Y_K}^*(X)<|X|$.

\subsection{Step 4: Matroid Union Theorem}\label{sec:proofunionmatroid}
In this step of the proof, we make use of the Matroid Union Theorem \cite{schrijver} to characterize an equivalent condition to (\ref{eq:condunionmatroid}). We start by stating the Matroid Union Theorem.
\begin{thm}(Matroid Union Theorem \cite[Chapter 42]{schrijver})\label{mut}
Let $M_{i}=(E_{i},\mathcal{I}_{i})$, $i\in[K]$, be $K$ matroids with rank functions $r_i(.)$. Then $\bigvee_{i=1}^K M_{i}$ is also a matroid with rank function
\begin{align}
r(U)=\min_{T\subseteq U} \left(|U\setminus T|+\sum_{i=1}^K r_{i}(T\cap E_{i})\right).
\end{align}
\end{thm}

Equipped with the above theorem, we can now state the following equivalence lemma.

\begin{lem}
Condition (\ref{eq:condtwo}) is equivalent to the following.
\begin{gather}
\forall X\subseteq[n]: |X|>R-\tau,\nonumber\\
\forall Y_i\subseteq[m_i], i\in[K] \text{ s.t. }\sum_{i=1}^K |Y_i|=|X|, \nonumber\\
\exists J\subseteq X:
\sum_{i=1}^K \dim (\mathcal{S}_{J\cup X^c}\cap \mathcal{B}_{i,*,Y_i})>|J|.\tag{C5}\label{eq:condthree}
\end{gather}
\end{lem}
\begin{proof}
If for a specific choice of $X$ and $Y_i$'s, (\ref{eq:matroiddefcondition}) is not satisfied (i.e., $\dim(\mathcal{S}_{X^c} \cap \mathcal{B}_{i,*,Y_i} )>0$ for some $i\in[K]$), then it is clear that both statements in~\eqref{eq:condtwo} and~\eqref{eq:condthree} hold. Hence, w.l.o.g. we assume that~\eqref{eq:matroiddefcondition} is satisfied. For any $X\subseteq [n], Y_i\subseteq[m_i], i\in[K] \text{ s.t. }\sum_{i=1}^K |Y_i|=|X|> R -\tau$ and (\ref{eq:matroiddefcondition}) is satisfied, we will get from Lemma~\ref{lem:matroidunion} that the condition in~\eqref{eq:condtwo} is equivalent to,
\begin{gather*}
r_{X,Y_1,...,Y_K}^*(X)<|X|\nonumber\\
\overset{(a)}{\Leftrightarrow}\underset{J\subseteq X}{\min} \left((|X|-|J|)+\sum_{i=1}^K (|Y_i|-\dim (\mathcal{S}_{J^c}\cap \mathcal{B}_{i,*,Y_i}))\right)<|X|\nonumber\\
\Leftrightarrow |X|+\sum_{i=1}^K |Y_i|-\underset{J\subseteq X}{\max} \left(|J|+\sum_{i=1}^K \dim (\mathcal{S}_{J^c}\cap \mathcal{B}_{i,*,Y_i})\right)<|X|\nonumber\\
\overset{(b)}{\Leftrightarrow} ~\max_{J\subseteq X} \bigg(|J|+\sum_{i=1}^K \dim (\mathcal{S}_{J^c}\cap \mathcal{B}_{i,*,Y_i})\bigg)>|X|\\
\Leftrightarrow ~\exists J\subseteq X : \sum_{i=1}^K \dim (\mathcal{S}_{J^c}\cap \mathcal{B}_{i,*,Y_i})>|X\setminus J|\\
\overset{(c)}{\Leftrightarrow} ~\exists J\subseteq X : \sum_{i=1}^K \dim (\mathcal{S}_{J\cup X^c}\cap \mathcal{B}_{i,*,Y_i})>|J|,
\end{gather*}
where (a) follows from (\ref{dualrank}) and Theorem \ref{mut}, (b) follows from the assumption that $\sum_{i=1}^K |Y_i|=|X|$ and (c) follows by changing $J$ to $X\setminus J$. This completes the proof.
\end{proof}
\subsection{Step 5: Hall's Marriage Theorem}\label{sec:prooffour}
In the final step of the proof, we prove the equivalence between~\eqref{eq:condthree} and~\eqref{eq:result} by constructing an appropriate bipartite graph. One partite set represents the column vectors of a carefully chosen set of bases for the subspaces $\mathcal{B}_{i,*,Y_i}$ and the other partite set has $n$ elements, each corresponding to a row. A column vector vertex is connected to a row vertex if and only if the vector has a non-zero entry in the corresponding row. We then use the following variation of Hall's marriage theorem to complete the proof by representing the final condition~\eqref{eq:result} in terms of a matching size on the constructed bipartite graph. 
\begin{thm}\label{thm:hull}
Let $G=(A\cup B,E)$ be a bipartite graph. $G$ has a matching of size $k$ if and only if the size of the neighboring set $|N(I)|\geq |I|-|B|+k$ for any $I\subseteq B$.
\end{thm}
We first show that \eqref{eq:result} $\Rightarrow$ \eqref{eq:condthree}. Assume that \eqref{eq:result} holds and fix $Y_i \subseteq [m_i], i\in[K]: \sum_{i=1}^K |Y_i|=R$. Let $J^* \subseteq [n]$ be such that $\sum_{i=1}^K \dim(\mathcal{S}_{J^*} \cap \mathcal{B}_{i,*,Y_i}) \geq |J^*|+\tau$. For each $i \in [K]$, let $c_1^{(i)},...,c_{|Y_i|}^{(i)}$ be $|Y_i|$ column vectors that form a basis for $\mathcal{B}_{i,*,Y_i}$ and a subset of these vectors form a basis for $\mathcal{S}_{J^*} \cap \mathcal{B}_{i,*,Y_i}$. Now, let $C_i=\left\{c_1^{(i)},...,c_{|Y_i|}^{(i)}\right\}$ and construct a bipartite graph $G=(A \cup C,E)$ where $A=\{1,2,\ldots,n\}$ and $C=\{C_1,\ldots,C_K\}$ is a multiset consisting of the elements of the sets $C_i, i\in [K]$. $A$ and $C$ are the two partite vertex sets of the graph $G$. For any $i\in [n]$ and $c\in C$, $(i,c)\in E$ if and only if the vector $c$ has a non-zero entry in the $i^{\textrm{th}}$ position, i.e., if $c=(c_1,\ldots,c_n)$ then $(a_i,c)\in E \Leftrightarrow c_i \neq 0$. 
Now, note that all the vertices in $C$ that correspond to vectors in $\mathcal{S}_{J^*}$ can only be connected to vertices in $A$ that correspond to the set $J^*$. Since $\sum_{i=1}^K \dim(\mathcal{S}_{J^*} \cap \mathcal{B}_{i,*,Y_i}) \geq |J^*|+\tau$, it follows that there is a subset of the partite set $C$ of size at least $|J^*|+\tau$ whose neigboring set has size at most $|J^*|$. It follows from Theorem~\ref{thm:hull} that there is no matching of size $R-\tau+1$ in the bipartite graph $G$. It follows that for any $X \subseteq [n]: |X| > R-\tau$ and any $K$ sets $Y'_i \subseteq C,i \in [K]: \sum_{i=1}^K |Y'_i|=|X|$, there is no matching between the vertices corresponding to the set $X$ in the partite set $A$ and the vertices in $\{Y'_1,\ldots,Y'_K\}$, and hence, there is a set $J \subseteq X$ such that the vertices corresponding to the set $X \backslash J$ are connected to less than $|X \backslash J|$ vertices in $\bigcup_{i=1}^K Y'_i$. It follows that for this choice of the set $J$, there are more than $|J|$ vertices in $\bigcup_{i=1}^K Y'_i$ that are not connected to any vertex in the set $X \backslash J$ in the partite set $A$. Now, if $\forall i\in[K], Y'_i \subseteq Y_i$, then $\sum_{i=1}^K \dim(\mathcal{S}_{J \cup \bar{X}} \cap \mathcal{B}_{i,*,Y'_i}) >|J|$. We now use the above argument to prove that \eqref{eq:condthree} holds as follows. For all $ Y'_i\subseteq[m_i], i\in[K]$ such that $\sum_{i=1}^K |Y'_i| > R-\tau$, we find $Y_i, i\in[K]$ such that $Y'_i \subseteq Y_i, \forall i \in [K]$ and $\sum_{i=1}^K |Y_i|=R$, and then use the above argument to show that, 
\begin{gather*}
\exists J \subseteq [n]: \sum_{i=1}^K \dim(\mathcal{S}_{J} \cap \mathcal{B}_{i,*,Y_i}) \geq |J|+\tau \Rightarrow\\ \forall X \subseteq [n]: |X| = \sum_{i=1}^K |Y'_i|, \\\exists J\subseteq X:
 \sum_{i=1}^K \dim(\mathcal{S}_{J \cup \bar{X}} \cap \mathcal{B}_{i,*,Y'_i}) >|J|,
\end{gather*}
and hence, \eqref{eq:condthree} follows.

We now show that \eqref{eq:condthree} $\Rightarrow$ \eqref{eq:result} by contradiction. Suppose that \eqref{eq:result} does not hold and fix $Y_i \subseteq [m_i], i\in[K]: \sum_{i=1}^K |Y_i|=R$ such that for any $J\subseteq [n]$, $\sum_{i=1}^K \dim(\mathcal{S}_{J} \cap \mathcal{B}_{i,*,Y_i}) <  |J|+\tau$. For each $i \in [K]$, let $c_1^{(i)},...,c_{|Y_i|}^{(i)}$ be $|Y_i|$ column vectors that form a basis for $\mathcal{B}_{i,*,Y_i}$ and define $C_i=\left\{c_1^{(i)},...,c_{|Y_i|}^{(i)}\right\}$. Now, we construct a bipartite graph $G=(A \cup C,E)$ where $A=\{1,2,\ldots,n\}$ and $C=\{C_1,\ldots,C_K\}$ is a multiset consisting of the elements of the sets $C_i, i\in [K]$. Also, for any $i \in [n], c\in C$, $(i,c) \in E \Leftrightarrow c_i \neq 0$.

From the selection of the sets $Y_i, i\in[K]$, we have that for any $I\subseteq C$, the neighboring set $N(I) \subseteq A$ has size $|N(I)| > |I|-\tau$. It follows from Theorem \ref{thm:hull} that there is a matching in $G$ of size $|C|-\tau+1$. Such a matching will include $R-\tau+1$ edges incident on $R-\tau+1$ nodes in $A$ (which we denote by $X$) and $R-\tau+1$ nodes in $C$ (which we denote by $Y'$). For each $i \in [K]$, let $Y'_i=\{j:c_j^{(i)}\in Y'\}$. Now, in order to show that \eqref{eq:condthree} is not true, it suffices to show that for the specific choice of $X$, $Y'_i, i\in[K]$ mentioned above, the following holds:
\begin{align}\label{contradict}
\forall J\subseteq X:
 \sum_{i=1}^K \dim(\mathcal{S}_{J \cup X^c} \cap \mathcal{B}_{i,*,Y'_i}) \leq |J|
\end{align}

Fix any $J\subseteq X$. We know that there exist $|X\setminus J|$ nodes in $Y'$ that are connected to the vertices corresponding to elements in $X\setminus J$ through the edges in the matching. This implies that there exist $|X\setminus J|$ columns in $\left[B_{1,*,Y'_1}~\ldots~B_{K,*,Y'_K}\right]$ which have at least one non-zero entry in $X\setminus J$ and therefore cannot belong to $\mathcal{S}_{J\cup X^c}$. This means that $\sum_{i=1}^K \dim(\mathcal{S}_{J \cup X^c} \cap \mathcal{B}_{i,*,Y'_i})$ cannot be greater than the dimension of the span of the remaining columns, which is at most $\left(\sum_{i=1}^K |Y'_i|\right)-|X\setminus J|=|J|$, verifying (\ref{contradict}). Hence \eqref{eq:condthree} does not hold.
\section{Application to Topological Interference Management}\label{sec:tim}
\noindent In this section, we study an application of Theorem~\ref{thm:main} to characterize the linear symmetric degrees of freedom for a class of topological interference management problems as defined next.

\subsection{Topological Interference Management: System Model and Problem Overview}
We consider $K$-user interference networks composed of $K$ transmitter nodes $\{\text{T}_i\}_{i=1}^K$ and $K$ receiver nodes $\{\text{D}_i\}_{i=1}^K$. Each transmitter $\text{T}_i$ intends to deliver a message $W_i\in\mathcal{W}_i$ to its corresponding receiver $\text{D}_i$. We assume that each receiver is subject to interference only from a specific subset of the other transmitters and the interference power that it receives from the other transmitters is below the noise level. This leads to a \emph{network topology} indicating the network interference pattern.

Each transmitter $\text{T}_i$ intends to send a vector $\mathbf{w}_i\in\mathbb{R}^{m_i}$ of $m_i$ symbols to its desired receiver $\text{D}_i$ over $n$ time slots. This message is encoded to the transmit vector $\mathbf{x}_i=B_i \mathbf{w}_i$, where $B_i$ denotes the linear beamforming \emph{precoding} matrix of transmitter $i$, which is of size $n\times m_i$. The received signal of receiver $j$ over $n$ time slots is given by,
\begin{align*}
\mathbf{y}_j=(\Lambda_{jj} B_j) \mathbf{w}_j+\sum_{i\in I_j} (\Lambda_{ij} B_i) \mathbf{w}_i+\mathbf{z}_j,
\end{align*}
where $I_j$ is the set of transmitters interfering at receiver $j$, $\Lambda_{ij}$ is the $n\times n$ diagonal matrix with the $k^{\textrm{th}}$ diagonal element being equal to the value of the channel coefficient between transmitter $i$ and receiver $j$ in time slot $k$, and $\mathbf{z}_j$ is the noise vector at receiver $j$ where each of its elements is an i.i.d. $\mathcal{N}(0,N)$ random variable, $N$ being the noise variance. The channel gain values are assumed to be identically distributed and drawn from a continuous distribution at each time slot. We assume that transmitters have no knowledge about the realization of the channel gains except for the topology of the network. However, the receivers have full channel state information. We refer to this assumption as no CSIT (channel state information at the transmitters) beyond topology. Each precoding matrix $B_i$ is an $n \times m_i$ matrix that can only depend on the knowledge of topology. 
At receiver $j$, the interference subspace denoted by $\mathcal{I}_j$ can be written as
\begin{align*}
\mathcal{I}_j=\bigcup_{i\in I_j} \mathsf{colspan}(\Lambda_{ij} B_i).
\end{align*}

In order to decode its desired symbols, receiver $j$ projects its received signal subspace given by $\mathsf{colspan} (\Lambda_{jj} B_j)$ onto the subspace orthogonal to $\mathcal{I}_j$, and its successful decoding condition can be expressed as
\begin{align}\label{eq:decodability}
\text{dim}\left(\mathsf{Proj}_{{\cal I}_{j}^c} \mathsf{colspan} \left(\Lambda_{jj} B_j\right)\right)=m_j.
\end{align}

If the above decodability condition is satisfied at all the receivers $\{\text{D}_j\}_{j=1}^K$ for almost all realizations of channel gains, then the linear degrees of freedom (LDoF) tuple $(\frac{m_1}{n},...,\frac{m_K}{n})$ is achievable under the aforementioned linear scheme. The linear symmetric degrees of freedom $\text{LDoF}_{sym}$ is defined as the supremum $d$ for which the LDoF tuple $(d,...,d)$ is achievable.

In this setting, the main goal is to characterize the linear symmetric degrees of freedom for general network topologies. There have been multiple attempts in the literature to resolve this problem. In particular, in \cite{jafar} it was shown that if the channel gains in the network remain constant for a sufficiently large time, then topological interference management is closely connected to the classical index coding problem (see, e.g., \cite{birk,baryossef,rouayheb}), and via this connection, a class of linear interference management schemes has been introduced which rely on the network topology knowledge to align the interference over time at unintended receivers. Furthermore, in~\cite{topology,topology_isit}, the authors considered a class of retransmission-based schemes that only allows transmitters to resend their symbols in order to assist with the neutralization of interference at the receivers. Besides their simplicity, these schemes are robust to channel variations over time and were shown to be optimal in terms of the symmetric degrees-of-freedom (DoF) in many classes of topologies via the outer bounds developed in \cite{topology,topology_isit}. More recently, the authors in \cite{avoidance_jafar} have characterized a necessary and sufficient condition for a topology under which interference avoidance can achieve the whole DoF region.

In the following two sections, we use our derived rank loss condition in order to address two problems in the topological interference management framework. First, in Section \ref{sec:best}, we characterize the ``best'' topologies for topological interference management; i.e., the topologies for which half linear symmetric DoF is achievable for the case of time-varying channels. For the case where the channel gains in the network are assumed to remain constant for a long-enough period of time, the necessary and sufficient condition on the network topology for achieving half linear symmetric DoF was characterized in \cite{jafar}. For the case of case of time-varying channels, a sufficient condition for the achievability of half linear symmetric DoF was derived in \cite{multialign}. In the next section, we close the gap by introducing a necessary and sufficient condition under which the half linear symmetric DoF is achievable for the case of time-varying channels (i.e., without requiring the channels to remain fixed for a long enough time). As for the second problem, in Section \ref{sec:nonoverlapping} we characterize the linear symmetric degrees of freedom for a class of network topologies with exclusive interference alignment sets. This helps to resolve the characterization of the linear symmetric DoF for a set of previously open problems considered in~\cite{topology}.

\subsection{Identifying the ``Best'' Topologies}\label{sec:best}
In this section, we use Theorem~\ref{thm:main} to identify network topologies where a linear symmetric DoF of $\frac{1}{2}$ is achievable. It is easy to see that in topologies where at least one transmitter is interfering at an undesired receiver, the symmetric DoF is upper bounded by $\frac{1}{2}$. Thus, the topologies in which half symmetric DoF is achievable represent the ``best'' topologies that one can hope for (from the degrees-of-freedom perspective). 
We start by making the following definition of \emph{reduced conflict graphs}.

\begin{defn}\label{redgraph}
The reduced conflict graph of a $K$-user interference network is a directed graph $G=(V,A)$ with $V=\{1,2,...,K\}$. As for the edges, vertex $i$ is connected to vertex $j$ (i.e., $(i,j)\in A$) if and only if $i \neq j$ and the following two conditions hold:
\begin{enumerate}
\item Transmitter $i$ is connected to receiver $j$.
\item $\exists s,k \in \{1,2,\ldots,K\}\backslash\{i\}, s\neq k$ such that both transmitter $i$ and transmitter $s$ are connected to receiver $k$. 
\end{enumerate}
\end{defn}

Clearly, the difference between the reduced conflict graph defined above and the regular conflict graph (used e.g. in \cite{topology}) is the additional condition 2 in Definition \ref{redgraph}, which implies that the set of edges in the reduced conflict graph is a subset of the edges in the regular conflict graph, hence the name ``reduced'' conflict graph.

Having the above definition, we characterize the best topologies in the following theorem.
\begin{thm}\label{thm:halfdof}
For a $K$-user interference network with arbitrary topology, half linear symmetric DoF can be achieved if and only if the reduced conflict graph of the network is bipartite. 
\end{thm}
\begin{remk}\label{remk:avoidance}
It is easy to see that the interference avoidance scheme, which schedules the users in an independent set of the regular conflict graph of the network, can achieve half symmetric DoF if and only if the regular conflict graph of the network is bipartite. Since the reduced conflict graph is the same as the regular conflict graph with the removal of edges that do not satisfy condition $2$ in Definition~\ref{redgraph}, there exist topologies where the reduced conflict graph is bipartite but the regular conflict graph is not bipartite. One such example topology is illustrated in Figure \ref{ex_gain} together with its regular and reduced conflict graphs in Figure \ref{conf_graph}.
\begin{figure}[h]
\centering
\begin{subfigure}{0.3\textwidth}
\centering
\includegraphics[trim = 2.68in 2.3in 5.4in 2.2in, clip,width=0.7\textwidth]{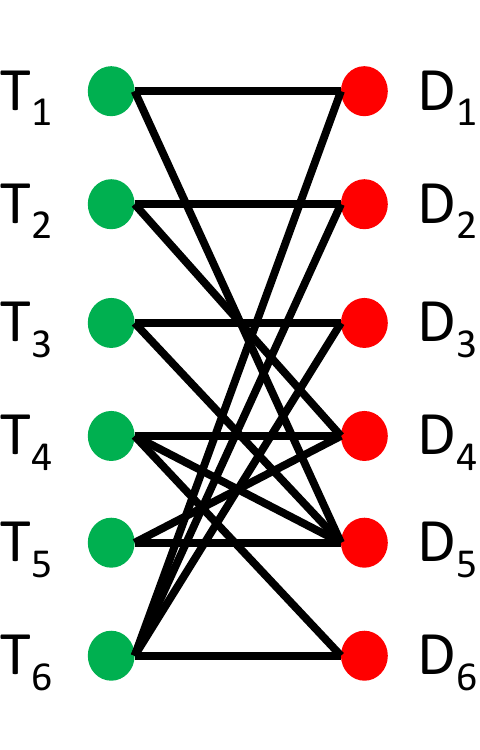}
\caption{}
\label{ex_gain}
\end{subfigure}
~
\begin{subfigure}{0.3\textwidth}
\centering
\includegraphics[trim = 5.4in 2.3in 1.9in 2.2in, clip,width=0.9\textwidth]{ex_top_avoidance}
\caption{}
\label{conf_graph}
\end{subfigure}
\caption{(a) A topology in which half symmteric DoF is achievable, but not by interference avoidance, and (b) the corresponding regular and reduced conflict graph. The dashed red lines exist in the conflict graph and are absent in the reduced conflict graph.}
\end{figure}
The black edges are the edges in the reduced conflict graph and the dashed red edges are the ones that exist in the regular conflict graph but are removed in the reduced conflict graph. It is clear that the reduced conflict graph is bipartite (with the two partite sets being $\{1,3,4\}$ and $\{2,5,6\}$), hence $d_{sym}=\frac{1}{2}$ is achievable in this topology. However the addition of the dashed red edges in the regular conflict graph removes the bipartiteness property of the graph and therefore  interference avoidance cannot achieve half symmetric DoF in this topology. In fact, since the chromatic number of the regular conflict graph in Figure \ref{conf_graph} is 3, interference avoidance can only achieve a symmetric DoF of $\frac{1}{3}$. The suboptimality of interference avoidance can also be seen as a result of the network topology in Figure \ref{ex_gain} not being chordal \cite{avoidance_jafar}.
\end{remk}




\begin{remk}
In \cite{jafar}, a linear scheme based on aligning the interference at unintended receivers was considered. The alignment scheme relies on the channel remaining constant for a long time (large coherence time). It was shown that half symmetric DoF is achievable if and only if there is no conflict between any two nodes that cause interference at a third receiver, i.e., there is no internal conflict within an alignment set. It is easy to show that the condition in Theorem \ref{thm:halfdof} implies the absence of internal conflicts but the opposite is not true; this is consistent with intuition as we make no assumption on the coherence time of the channel. Interestingly, it is shown in~\cite{ICC} that the topologies for which interference alignment can achieve half symmetric DoF while retransmission cannot achieve it comprise a negligible fraction of possible topologies in a heterogeneous network scenario of practical interest.
\end{remk}

\begin{proof}[Proof of Theorem \ref{thm:halfdof}] 
We first prove the converse. In order to achieve $\frac{n}{2}$ symmetric DoF, there has to be a sequence $(\epsilon_n, n\in{\bf Z}^+)$, such that $\epsilon_n \rightarrow 0$ and $\frac{n}{2}-\epsilon_n$ symmetric DoF is achievable by coding over $n$ time slots. Consider some generic receiver $l$ in the network subject to interference from transmitters $\{\text{T}_i\}_{i=1}^K$. Then the decodability condition is equivalent to the dimension of the interference subspace being at most $\frac{n}{2}+\epsilon_n$ almost surely; i.e.,
\begin{align}\label{eq:halfdofrankloss}
\emph{rank} ([\Lambda_{1l} B_1 ~~\Lambda_{2l} B_2~~ ... ~~ \Lambda_{Kl} B_K])\overset{a.s.}{\leq} \frac{n}{2}+\epsilon_n.
\end{align}

Then, we know from Theorem~\ref{thm:main} that \eqref{eq:halfdofrankloss} implies the following for $K \geq 2$ and $m_i=\frac{n}{2}-\epsilon_n, \forall i\in[K]$.
\begin{align}\label{eq:halfcondone}
\forall Y_i\subseteq\left[\frac{n}{2}\right], i \in [K] \emph{ s.t. }\sum_{i=1}^K |Y_i|=\min\left(n, K\left(\frac{n}{2}-\epsilon_n\right)\right),
\exists J \subseteq [n]: \sum_{i=1}^K \dim(\mathcal{S}_{J} \cap \mathcal{B}_{i,*,Y_i}) \geq |J|+\frac{n}{2} - 3\epsilon_n.
\end{align}

Now, we can state the following lemma that simplifies the condition~\eqref{eq:halfcondone} resulting from Theorem~\ref{thm:main} for this particular setting.
\begin{lem}\label{lem:halfdofcondition}
For the case where $K \geq 2, m_i=\frac{n}{2}-\epsilon_n, \forall i\in[K]$, the statement in~\eqref{eq:halfcondone} implies the following.
\begin{gather}\label{eq:halfcondtwo}
\exists \tilde{J} \subseteq [n]: \left|\tilde{J}\right|\geq\frac{n}{2}-c K \epsilon_n, \dim\left(\cap_{i=1}^K \mathcal{B}_i \cap \mathcal{S}_{\tilde{J}}\right) \geq \frac{n}{2}-c K \epsilon_n,
\end{gather}
where $c$ is a constant integer.
\end{lem}
\begin{proof}
Assume that~\eqref{eq:halfcondone} holds, and set $Y_1=Y_2=\left[\frac{n}{2}-\epsilon_n\right]$. It then follows that,
\begin{equation}\label{eq:halfproofone}
\exists J\subseteq[n]: \dim(\mathcal{S}_{J} \cap \mathcal{B}_{1})+\dim(\mathcal{S}_{J} \cap \mathcal{B}_{2}) \geq |J| + \frac{n}{2}-5 \epsilon_n.
\end{equation}

Since $\dim(\mathcal{S}_{J} \cap \mathcal{B}_{i})\leq \min\left(|J|,m_i\right), \forall J\subseteq[n], i\in[K]$, it follows from~\eqref{eq:halfproofone} that there exists $J\subseteq[n], |J| \geq \frac{n}{2}-5\epsilon_n, \dim\left(\mathcal{B}_1 \cap \mathcal{B}_2 \cap S_J\right) \geq \frac{n}{2}-9\epsilon_n$; let this set $J$ be called $J_{1,2}$. Similarly, by taking $Y_2=Y_3=\left[\frac{n}{2}-\epsilon_n\right]$, there exists $J \subseteq[n],|J| \geq \frac{n}{2}-5\epsilon_n, \dim\left(\mathcal{B}_2 \cap \mathcal{B}_3 \cap S_J\right) \geq \frac{n}{2}-9\epsilon_n$; call this set $J$ as $J_{2,3}$. Since $\dim\left(\mathcal{B}_2 \cap S_{J_{1,2}}\right) \geq \frac{n}{2}-9\epsilon_n$ and $\dim\left(\mathcal{B}_2 \cap S_{J_{2.3}}\right) \geq \frac{n}{2}-9\epsilon_n$, then $\left|J_{1,2} \cap J_{2,3}\right| \geq \frac{n}{2}-17\epsilon_n$. Also, let $J_{1,2,3}=J_{1,2} \cap J_{2,3}$, then $\dim\left(\mathcal{B}_1 \cap \mathcal{B}_2 \cap \mathcal{B}_3 \cap S_{J_{1,2,3}}\right) \geq \frac{n}{2}-17 \epsilon_n$. Proceeding in the same way, we can show that there exists a constant integer $c$ such that if we let $\tilde{J}=J_{1,2}\cap J_{2,3} \cap \ldots \cap J_{K-1,K}$, then $\left|\tilde{J}\right| \geq c K \epsilon_n$ and $\dim\left(\mathcal{B}_1 \cap \ldots \cap \mathcal{B}_K \cap S_{\tilde{J}}\right) \geq c K \epsilon_n$, and hence,~\eqref{eq:halfcondtwo} follows.
\end{proof}

We use the result of Lemma~\ref{lem:halfdofcondition} together with the following lemma to prove the converse of Theorem~\ref{thm:halfdof}.
\begin{lem}\label{lem:conflict}
If transmitter $i$ is connected to receiver $k$, $k \neq i$ and there exist $J_i, J_k \subseteq [n]$ such that $\dim\left(\mathcal{B}_i \cap S_{J_i}\right) \geq \frac{n}{2}-cK\epsilon_n$ and $\dim\left(\mathcal{B}_k \cap S_{J_k}\right) \geq \frac{n}{2}-c K \epsilon_n$, then the decodability condition of~\eqref{eq:decodability} is satisfied at receiver $k$ only if  $|J_i \cap J_k| \leq 2 c K \epsilon_n$.
\end{lem}
\begin{proof}
Since $\dim\left(S_{J_i} \cap \mathcal{B}_i\right) \geq \frac{n}{2}-cK\epsilon_n$, then $\dim(S_{J_i} \cap \mathsf{colspan}(\Lambda_{ik}B_i)) \geq \frac{n}{2}-cK\epsilon_n$ almost surely. Similarly, if $\dim\left(S_{J_k} \cap \mathcal{B}_k\right) \geq \frac{n}{2}-cK\epsilon_n$, then $\dim(S_{J_k} \cap \mathsf{colspan}(\Lambda_{kk}B_k)) \geq \frac{n}{2}-cK\epsilon_n$ almost surely. If $|J_i \cap J_k| > 2cK\epsilon_n$, then $\dim(\mathsf{colspan}(\Lambda_{ik}B_i) \cap \mathsf{colspan}(\Lambda_{kk}B_k)) >0$, and hence, it follows that the dimension of the projection of the received signal at receiver $k$ on the complement of the interference subspace is less than $\dim(\mathcal{B}_k)$, i.e., $\mathsf{dim}\left(\mathsf{Proj}_{{\cal I}_{k}^c} \mathsf{colspan} \left(\Lambda_{kk} B_k\right)\right) < \frac{n}{2}-\epsilon_n$, and therefore, violating~\eqref{eq:decodability}.
\end{proof}

Now, assume that the reduced conflict graph $G$ is not bipartite, then its chromatic number is at least $3$. In this case, we know from Lemma~\ref{lem:halfdofcondition} and Lemma~\ref{lem:conflict} that there exist three users (say users $1$, $2$ and $3$) and three sets $J_i \subseteq [n], i\in\{1,2,3\}$ such that the following holds.
\begin{equation}\label{eq:convcondone}
|J_i| \geq \frac{n}{2} - cK\epsilon_n, \forall i\in\{1,2,3\},
\end{equation}
\begin{equation}\label{eq:convcondtwo}
|J_1 \cap J_2| \leq 2cK\epsilon_n, |J_2 \cap J_3| \leq 2cK\epsilon_n, |J_1 \cap J_3| \leq 2cK\epsilon_n.
\end{equation}
 It is easy to see that if $7cK\epsilon_n < \frac{n}{2}$, then one of the conditions in~\eqref{eq:convcondone} and~\eqref{eq:convcondtwo} is violated. It then follows that it cannot be the case that $\epsilon_n \rightarrow 0$, and hence, $\frac{1}{2}$ symmetric DoF cannot be achieved.

Conversely, if $G$ is bipartite, we show a linear coding scheme achieving half symmetric DoF. Consider a coding scheme over two time slots and suppose that each transmitter uses a point-to-point capacity achieving code and whenever it is activated in any of the two time slots, it transmits the selected codeword and otherwise it remains silent. The activation of transmitters is determined by the graph $G$ as follows. Let $P_1$ and $P_2$ be the two partite sets constituting $G$. If vertex $i$ has no outgoing edges, then transmitter $i$ is activated in both time slots. For any remaining user $j$ (corresponding to the nodes that have least one outgoing edge in $G$), if vertex $j$ is in $P_k$, $k\in\{1,2\}$, then transmitter $j$ is active in time slot $k$ and inactive in the other time slot. As an example, for the 6-user topology of Figure \ref{ex_gain}, the retransmission pattern can be written as follows,
\begin{align}\label{code}
\setlength{\arraycolsep}{4pt}
\begin{bmatrix}
B_1 & B_2 & B_3 & B_4 & B_5 & B_6 
\end{bmatrix}
=
\begin{bmatrix}
0 & 1 & 0 & 0 & 1 & 1\\
1 & 0 & 1 & 1 & 0 & 1
\end{bmatrix}.
\end{align}

Each column in~\eqref{code} corresponds to a user and each row corresponds to a time slot. For instance, transmitter 2 sends its codeword in time slot 1 and remains silent in time slot 2, whereas transmitter 6 repeats its codeword in both time slots.

We now show that successful decoding is possible for almost all realizations of the channel coefficients, and hence, $1$ DoF is achieved for each user over two time slots. For each receiver $i$ with more than two interfering links, transmitter $i$ is activated in a time slot where all interfering transmitters are silent, and hence, the successful decoding condition of~\eqref{eq:decodability} is guaranteed. In the example of Figure~\ref{ex_gain}, receivers $4$ and $5$ have two interfering links. Transmitters $2$ and $5$ are interfering at receiver $4$, and hence, in~\eqref{code}, transmitters $2$ and $5$ are silent in the second time slot where transmitter $4$ is active. Similarly, transmitters $3$ and $4$ are interfering at receiver $5$, and hence, both are silent in the first time slot where transmitter $5$ is active. For each receiver $i$ with one interfering link, it is either the case that transmitter $i$ is active in a time slot for which the interfering transmitter is silent (as is the case for receiver $6$ in the example), or it is the case that at least one of transmitter $i$ and the interfering transmitter is active in both time slots (as is the case for receivers $1$, $2$ and $3$ in the example); in both cases, the condition in~\eqref{eq:decodability} is satisfied. This completes the proof.
\end{proof}

\subsection{Characterizing the Linear Symmetric DoF for Network Topologies with Exclusive Alignment Sets}\label{sec:nonoverlapping}

In this section, we characterize the linear symmetric degrees of freedom for a broader class of topologies; namely the topologies with ``exclusive alignment sets'', to be defined shortly. This generalizes the result that we proved in the previous section for half linear symmetric DoF.

Consider a coding scheme achieving a linear symmetric DoF $d$ over $n$ time slots. At each receiver $j$, the decodability condition~\eqref{eq:decodability} implies that
$\text{dim} \left({\cal I}_j\right) \leq n\left(1-d\right)$.
Using Theorem~\ref{thm:main}, we obtain the following equivalent condition on the design of beamforming matrices corresponding to interfering signals,
\begin{gather}
\forall Y_i\subseteq[nd], i \in I_j \emph{ s.t. }\sum_{i \in I_j} |Y_i|=\min\left(nd|I_j|,n\right)\nonumber,\\
\exists J \subseteq [n]: \sum_{i \in I_j} \dim(\mathcal{S}_{J} \cap \mathcal{B}_{i,*,Y_i}) \geq |J|+\tau,\label{eq:finalalignment}
\end{gather}
where $\tau=n\left(3d-1\right)$. Using the condition in~\eqref{eq:finalalignment} to reach a converse for the achievable linear symmetric DoF of arbitrary network topologies is a difficult problem, as it is not clear what the required number of time slots $n$ is to achieve $d_{\text{sym}}$. Further, for each value $n$, reducing the complexity of the search for the optimal design of the beamforming matrices by converting it to a problem of combinatorial optimization does not seem straightforward. However, under the restriction to a certain class of topologies, the task becomes easier as it reduces to a problem independent of the value of $n$, and can be described directly in terms of the interference conflict pattern between network users. We introduce the properties of the considered network topologies through the interference sets $I_j, j \in [K]$.
 We consider topologies where the following properties hold,
\begin{itemize}
\item {\bf (P1) Maximum Degree:} For all $j \in [K]$, $|I_j| \leq 2$,
\item {\bf (P2) Exclusive Alignment Sets:} For all $j,k \in [K]$ such that $\max\left(|I_j|,|I_k|\right)=2$, $I_j \cap I_k = \phi$.  
\end{itemize}
The property (P1) simplifies the problem because for any network with at least one interfering link, $d_{\text{sym}} \leq \frac{n}{2}$. (P1) then implies that $nd|I_j| \leq n$. Therefore,~\eqref{eq:finalalignment} reduces to,
\begin{align}\label{eq:reducedalignment}
\exists J \subseteq [n]: \sum_{i \in I_j} \dim(\mathcal{S}_{J} \cap \mathcal{B}_{i}) \geq |J|+\tau.
\end{align}

Moreover, (P2) simplifies the problem as in this case we can conclude from~\eqref{eq:reducedalignment} that there is no loss in generality in assuming that $\mathcal{S}_J \subseteq \mathcal{B}_i, \forall i \in I_j$ (see Lemma~\ref{lem:non_overlapping} below). The problem then becomes that of finding sparse subspaces for each interference set of size $2$ such that the subspaces corresponding to conflicting interference sets do not overlap. This is captured through the chromatic number of a reduced conflict graph that captures only conflicts between interference sets.

Having the aforementioned properties, we can now state the result on the linear symmetric DoF for the considered class of topologies. We call any network topology with at least one interference link an interference network topology.
\begin{thm}\label{thm:dsym}
For any interference network topology satisfying (P1) and (P2), the linear symmetric DoF is given by,
\begin{equation}\label{eq:ldof}
\emph{LDoF}_{\text{sym}} = \min\left(\frac{1}{2},\frac{\chi(G)+1}{3\chi(G)}\right),
\end{equation}
where $G$ is the reduced conflict graph of the topology (as defined in Definition \ref{redgraph}) and $\chi(.)$ denotes the chromatic number.
\end{thm}
\begin{remk}
In~\cite{topology}, several examples of topologies have been discussed for which the converse for $d_{\text{sym}}$ has remained open. Using Theorem \ref{thm:dsym} , we can now characterize the linear symmetric DoF of all those topologies which satisfy (P1) and (P2). In Figure~\ref{fig:example}, we plot one example where the chromatic number of the reduced conflict graph is 3. The symmetric DoF of $\frac{4}{9}$ is achievable through the structured repetition coding scheme introduced in~\cite{topology}. Here, the converse follows from Theorem~\ref{thm:dsym}.
\end{remk}

\begin{figure}[h]
\centering
\begin{subfigure}{0.3\textwidth}
\centering
\includegraphics[trim = 1.48in 2.3in 6.3in 2.5in, clip,width=.8\textwidth]{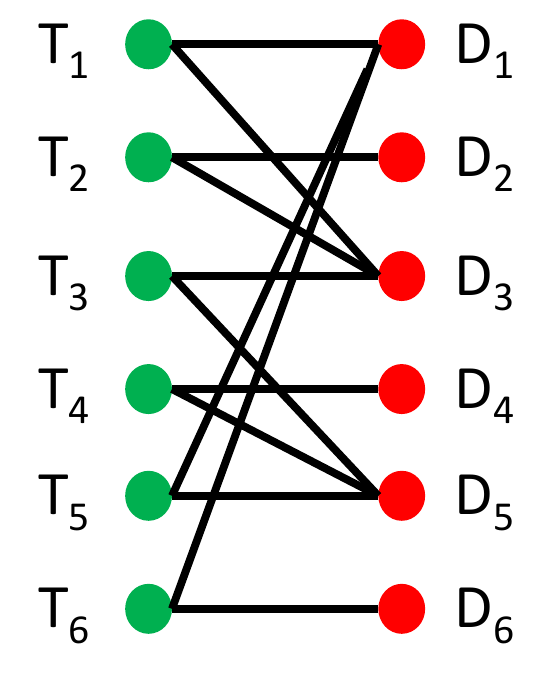}
\caption{}
\label{fig:example_topology}
\end{subfigure}
~
\begin{subfigure}{0.3\textwidth}
\centering
\includegraphics[trim = 5.6in 2.2in 1.0in 2.1in, clip,width=\textwidth]{4_9_nonoverlapping}
\caption{}
\label{conf_graph_4_9}
\end{subfigure}
\caption{(a) A topology in which the symmetric DoF $d_{\text{sym}}=\frac{4}{9}$  and (b) the corresponding reduced conflict graph with chromatic number $\chi(G)=3$.}
\label{fig:example}
\end{figure}
\begin{proof}[Proof of Theorem \ref{thm:dsym}] 
We know from~\cite{jafar} that for any interference network topology, $d_{\text{sym}} \leq \frac{1}{2}$. For the case where the reduced conflict graph has no edges, i.e., $\chi(G) = 1$, $d_{\text{sym}} = \frac{1}{2}$ is achievable by having each of the column vectors in each matrix $B_i, i\in[K]$ to have no zero entries. Hence, we consider the case where $\chi(G) \geq 2$ and show that $\text{LDoF}_{\text{sym}}=\frac{\chi(G)+1}{3\chi(G)}$ in this case.
We first show that for any receiver with two or more interfering signals, the sparse subspace $\mathcal{S}_J$ of~\eqref{eq:result} is \emph{fully occupied} by the interference, almost surely. More precisely, we prove the following corollary of the equivalent condition of Theorem~\ref{thm:main}.
\begin{lem}\label{lem:minimal}
If~\eqref{eq:result} holds, then for a minimal set $J$ satisfying~\eqref{eq:result}, $\mathcal{S}_J\overset{a.s.}{\subseteq} \mathcal{B}_D$. More precisely,
\begin{gather}
\exists Y_i\subseteq[m_i], i \in [K] \text{ s.t. }\sum_{i=1}^K |Y_i|=\min \left(\sum_{i=1}^K m_i,n\right),\nonumber\\ J\subseteq [n]: \sum_{i=1}^K \dim(\mathcal{S}_{J} \cap \mathcal{B}_{i,*,Y_i}) \geq |J|+x, x \geq \tau, \nonumber\\
\bigwedge \nexists L \subset J: \sum_{i=1}^K \dim(\mathcal{S}_{L} \cap \mathcal{B}_{i,*,Y_i}) \geq |L|+x, \nonumber\\ \Rightarrow \mathcal{S}_J\overset{a.s.}{\subseteq}\mathsf{colspan}\left(\left[\Lambda_1 B_1 ~ \ldots ~ \Lambda_K B_K\right]\right).\label{eq:minimal}
\end{gather}
\begin{proof}
Let $J^*$ be a set satisfying the condition in~\eqref{eq:minimal}. For each $i\in[K]$, let $c_1^{(i)},\ldots,c_{n_i}^{(i)}$ be $n_i$ vectors that form a basis for $\mathcal{S}_{J^*} \cap \mathcal{B}_{i,*,Y_i}$, where $\sum_{i=1}^K n_i = |J^*|+x$ and let $C_i=\left\{c_1^{(i)},\ldots,c_{n_i}^{(i)}\right\}$. Let $C=\{C_1,\ldots,C_K\}$ be the multiset consisting of the elements of $C_i, i\in[K]$ and let $G=(J^* \cup C, E)$ be the bipartite graph whose left partite set consists of vertices corresponding to elements in $J^*$ and right partite consists of vertices corresponding to the elements in $C$, and $\forall i\in J, c \in C, (i,c) \Leftrightarrow c_i \neq 0$. 
Since $\nexists L \subseteq J: \sum_{i=1}^K \dim(\mathcal{S}_{L} \cap \mathcal{B}_{i,*,Y_i}) \geq |L|+x+1$, we know that for any subset of vertices $I \subseteq C$, the neighboring set $N(I)$ satisfies the condition $|N(I)| \geq |I|-x$. It follows from Theorem~\ref{thm:hull} that there is a matching in $G$ of size $|C|-x$. Also, since $|J^*|=|C|-x$, we know that there is a matching in $G$ covering all elements of the left partite set. For each $i \in [K]$, let $c_1^{*(i)},\ldots,c_{n^*_i}^{*(i)}$ be the elements of $\{c_1^{(i)},\ldots,c_{n_i}^{(i)}\}$ in the matching in the right partite set $C$, where $\sum_{i=1}^K n^*_i=|J^*|$. 
For each $i\in[K]$, let $C_i^*=\left\{c_1^{*(i)},\ldots,c_{n^*_i}^{*(i)}\right\}$. Let $C^*$ be the multiset $C^*=\{C_1^*,\ldots,C_K^*\}$, and consider the bipartite graph $G^*=(J^* \cup C^*, E^*)$, where $\forall i\in J^*, c\in C^*, (i,b) \in E^* \Leftrightarrow c_i \neq 0$. Since $G^*$ has a perfect matching, it follows that $\nexists J \subseteq J^*: \sum_{i=1}^K \dim(\mathcal{S}_{J} \cap \mathcal{B}_{i,*,C_i^*}) > |J|$, and hence, from~\eqref{eq:result} we know that $[\Lambda_1 B_{1,*,C_1^*} ~ \ldots ~ \ldots ~ \Lambda_K B_{K,*,C_K^*}]$ is full rank almost surely. Now, since $\mathsf{colspan}\left(\left[\Lambda_1 B_{1,*,C_1^*} ~ \ldots ~ \Lambda_K B_{K,*,C_K^*}\right]\right)$ has $|J^*|$ linearly independent column vectors almost surely, it follows that $\mathcal{S}_{J^*} \overset{a.s.}{\subseteq} \mathsf{colspan}\left(\left[\Lambda_1 B_{1,*,C_1^*} ~ \ldots ~ \Lambda_K B_{K,*,C_K^*}\right]\right)$, and hence, $\mathcal{S}_{J^*} \overset{a.s.}{\subseteq} \mathcal{B}_D$.
\end{proof}
\end{lem}

From the condition in~\eqref{eq:result}, Lemma~\ref{lem:minimal} and the decodability condition~\eqref{eq:decodability}, we obtain the following condition,
\begin{gather}
\forall r \in [K]: I_r=\{r_1,r_2\}, \exists J_r \subseteq [n]: \nonumber\\\dim(S_{J_r} \cap \mathcal{B}_{r1}) + \dim(S_{J_r} \cap \mathcal{B}_{r2}) \geq |J_r|+\tau, \nonumber\\\dim(S_{J_r} \cap \mathcal{B}_r)=0.\label{eq:alignment}
\end{gather}

We now use the following lemma to restrict our attention to a simpler condition for the considered class of topologies. The proof of the lemma is in Appendix~\ref{app:non_overlapping}.
\begin{lem}\label{lem:non_overlapping}
For the case where alignment sets are exclusive, if there exist $B_1,\ldots,B_K$ such that \eqref{eq:alignment} holds, then there exist $B_1,\ldots,B_K$ such that the following condition holds,
\begin{gather}
\forall r \in [K]: I_r=\{r_1,r_2\}, \exists J_r \subset [n]: \nonumber\\
|J_r|=\tau, \mathcal{S}_{J_r} \subseteq \mathcal{B}_{r1} \cap \mathcal{B}_{r2}, \dim(\mathcal{S}_{J_r} \cap \mathcal{B}_r)=0.\label{eq:newalignment}
\end{gather}
\end{lem}

We now complete the proof by arguing that there exists a design of the beamforming matrices that satisfies~\eqref{eq:newalignment} if and only if $\tau \leq \frac{n}{\chi(G)}$. This follows directly by observing that the decodability condition in~\eqref{eq:newalignment} $\left(\dim(\mathcal{S}_{J_r} \cap \mathcal{B}_r)=0\right)$ is satisfiable if and only if there is no overlap between sparse subspaces corresponding to conflicting interference sets; more precisely,
\begin{gather}
r,d \in [K]: r\in I_d, |I_r|=|I_d|=2 \nonumber\\
\Rightarrow J_r \cap J_d = \phi.\label{eq:newdecode}
\end{gather}
\end{proof}
In order to prove the converse part of Theorem~\ref{thm:dsym}, we showed that the property (P2) allows us to assume that for the optimal coding scheme, whenever the interference alignment condition in~\eqref{eq:reducedalignment} is satisfied for an alignment set, it is the case that $|J|=\tau$ and $S_J \subseteq {\cal B}_i, \forall i\in I_j$. We end this section with two remarks on the generalization of the converse for arbitrary network topologies; in particular, topologies that do not satisfy the property (P2).

\begin{remk}
If the exclusive alignment set property (P2) is not satisfied, then the symmetric linear degrees of freedom can be larger than the value in the statement of Theorem~\ref{thm:dsym}. Consider the $9$-user network depicted in Figure~\ref{fig:ex_top} and note that transmitter $7$ is part of the alignment set at receiver $9$, and is also causing interference at receivers $2$, $4$ and $6$.
While $\frac{\chi(G)+1}{3\chi(G)}=\frac{5}{12}$, we show a scheme that achieves a symmetric DoF of $\frac{3}{7}$. This value of the symmetric DoF can be achieved by using a symbol extension $n=7$ and assigning the sparse subspaces as specified in Figure~\ref{fig:design}.
We explain why almost surely, the decodability condition of~\eqref{eq:decodability} is met at all receivers with one or more interfering transmitter. Since the sparse subspace for the alignment set at receiver $9$ intersects with any of $S_{\{1,2\}}, S_{\{3,4\}}$ and $S_{\{5,6\}}$ in a subspace of dimension at most $1$, the two column vectors in ${\cal B}_7 \cap S_{\{1,3,5\}}$ will almost surely lie outside any of these subspaces, and hence, the decodability conditions at receivers $2$, $4$ and $6$ will be met. Since $S_{\{1,2\}}$, $S_{\{3,4\}}$ and $S_{\{5,6\}}$ do not overlap, the desired signal subspace does not overlap with the interference subspace at each of receivers $1$, $3$ and $5$. Finally, at each of the receivers $1$, $3$, $5$ and $9$, the rank of the interference subspace will be at most $4$ almost surely because~\eqref{eq:reducedalignment} is satisfied with $\tau=2$.  
\end{remk}

\begin{figure}[htb]
\centering
\includegraphics[width=0.5\columnwidth]{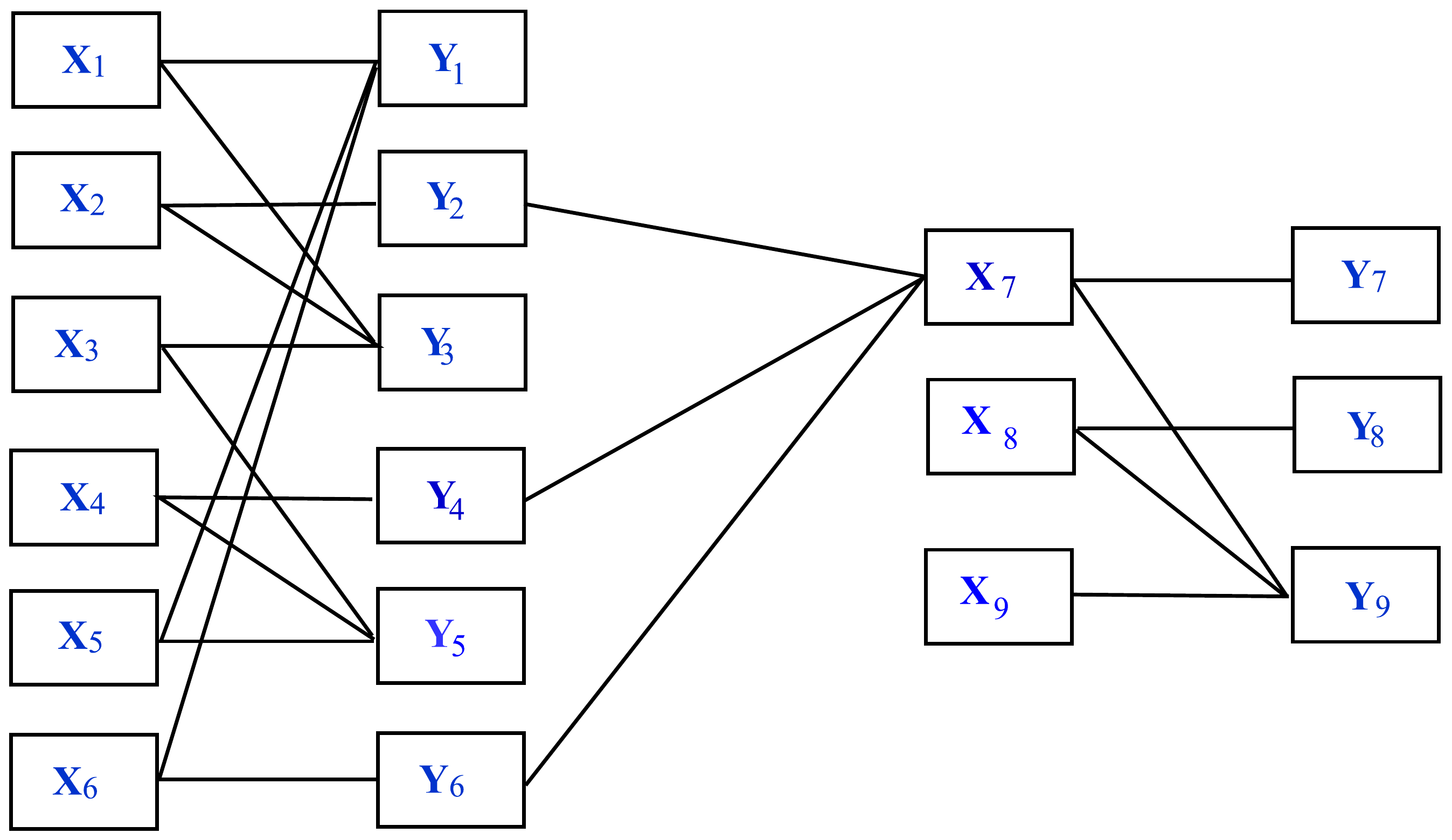}
\caption{An example of a 9-user network where $\frac{3}{7}$ symmetric DoF are achieved but $\frac{\chi(G)+1}{3\chi(G)}=\frac{5}{12}$.}
\label{fig:ex_top}
\end{figure}

\begin{figure}[htb]
\centering
\includegraphics[width=0.8\columnwidth]{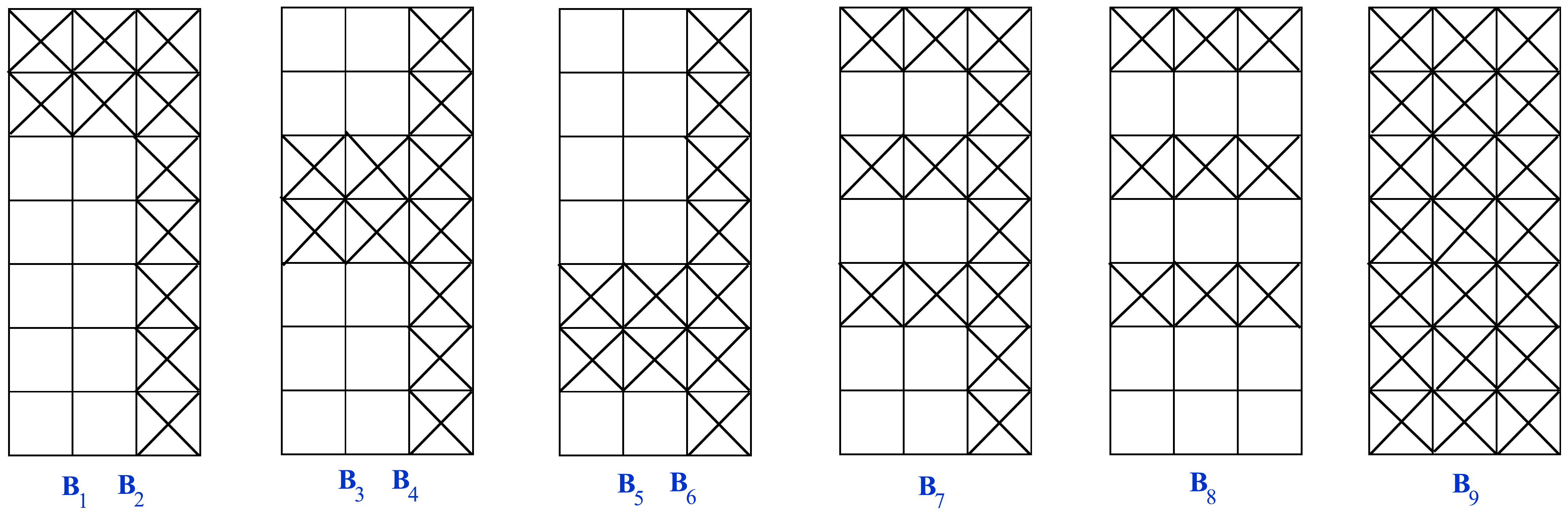}
\caption{The design of beamforming matrices for the example topology of Figure~\ref{fig:ex_top}. Each figure designates the places of the non-zero entries in the corresponding matrices by crossed squares. Each non-zero entry is drawn independently from a continuous distribution, and successful decoding is guaranteed almost surely.}
\label{fig:design}
\end{figure}

\begin{remk}
We believe that the problem of deriving a converse on the linear symmetric DoF for general network topologies can be simplified by describing each beamforming matrix $B_i$ by the number of vectors it has in each possible sparse subspace. For any set $J \subseteq [n]$, we let $\mu_i(J)$ be the number of vectors that $\mathcal{B}_i$ has in the sparse subspace $S_J$. The numbers $\left\{\mu_i(J), J \subseteq [n]\right\}$ have to satisfy the following constraints for each $i \in [K]$:
$\mu_i([n])=nd$, $\mu_i(J_1\cup J_2)\geq \mu_i(J_1)+\mu_i(J_2)-\mu_i(J_1\cap J_2),~  \forall J_1,J_2\subseteq[n]$, $\mu_i(J_1)\geq \mu_i(J_2)-\left(|J_2|-|J_1|\right). ~ \forall J_1\subseteq J_2\subseteq [n]$.
We can then derive an upper bound on the linear symmetric DoF for each value of $n$ through the solution of an optimization problem that relies on this description for beamforming matrices.
We believe that the upper bound obtained through this method is tight for any network topology. The key question here is to validate whether it is true that for any set of numbers $\left\{\mu_i(J), J\subseteq [n]\right\}$ satisfying the stated conditions, there exists an $n \times d$ matrix $B_i$ such that the number of vectors that $\mathcal{B}_i$ occupies in any sparse subspace $S_J, J \subseteq [n]$ is given by $\mu_i(J)$. 

\end{remk}

\section{Conclusion}\label{sec:conclusion}
We characterized necessary and sufficient conditions for almost sure rank loss of a concatenation of full rank matrices with randomly scaled rows. The characterized condition is in terms of the combinatorial structure of the individual matrices with respect to sparse column subspaces. We showed that an almost sure rank loss by a factor $\tau$ is possible if and only if the total number of column vectors in the ensemble in a sparse subspace exceeds the size of the sparse subspace by the same factor $\tau$. 
We then used the result to characterize the linear symmetric DoF for a class of topological interference management problems that was previously studied in~\cite{topology}. In particular, we could identify necessary and sufficient conditions on the network topology, under which, a linear symmetric DoF of $\frac{1}{2}$ is achievable. Further, we characterized the linear symmetric DoF for any network topology with a maximum receiver degree of $3$ and exlusive alignment sets. In general, our result solves an underlying fundamental problem in topological interference management; we are considering for future work how it can be used as a cornerstone to characterize the linear degrees of freedom region for arbitrary network topologies.

\appendices


\section{Proof of Claim~\ref{matroid}}\label{sec:matroid_proof}
Without loss of generality, we consider the case of $i=1$. We first show that $M_{1,X,Y_1}$ is a matroid. To this end, we need to prove the following properties.
\begin{itemize}
\item If $I\subseteq J$ and $J\in \mathcal{I}_{1,X,Y_1}$, then $I\in \mathcal{I}_{1,X,Y_1}$: This is clear since $\mathcal{S}_{I\cup X^c}\subseteq \mathcal{S}_{J\cup X^c}$ and we therefore have
\begin{gather*}
\dim (\mathcal{S}_{I\cup X^c}\cap \mathcal{B}_{1,*,Y_1}) \leq \dim (\mathcal{S}_{J\cup X^c}\cap \mathcal{B}_{1,*,Y_1}) =0 \nonumber\\
\Rightarrow \dim (\mathcal{S}_{I\cup X^c}\cap \mathcal{B}_{1,*,Y_1}) =0.
\end{gather*}

\item If $I\in\mathcal{I}_{1,X,Y_1}$, $J\in\mathcal{I}_{1,X,Y_1}$ and $|I|<|J|$, then $\exists j\in J \setminus I$ s.t. $I\cup \{j\}\in\mathcal{I}_{1,X,Y_1}$:
Assume $I=\{i_1,...,i_k\}$ and $J=\{i_1,...,i_{k'},j_1,...,j_l\}$ where we have $k'\leq k$ and $k'+l>k$. Using this notation, we have that $J\setminus I=\{j_1,...,j_l\}$. Suppose that this property is not true. This implies that adding any of the $e_{j_n}$'s to $I$ will force it to lie outside the set $\mathcal{I}_{1,X,Y_1}$ (where $e_{j_n}$ denotes the $j_n^{th}$ standard basis vector). If $ X^c=\{x_1,...,x_m\}$, this is equivalent to the fact that there exist coefficients $\lambda_{ij}$ ($i\in [l],j\in[k]$), $\beta_{ij}$ ($i\in[l],j\in[m]$) and $\mu_i\neq 0$ ($i\in[l]$) such that
\begin{align}
\lambda_{11}e_{i_1}+...+\lambda_{1k}e_{i_k}+\beta_{11}e_{x_1}+...+\beta_{1m} e_{x_m}+  mu_{1}e_{j_1}&=v_1\label{coeffs1}\\
\lambda_{21}e_{i_1}+...+\lambda_{2k}e_{i_k}+\beta_{21}e_{x_1}+...+\beta_{2m} e_{x_m}+mu_{2}e_{j_2}&=v_2\\
\vdots~~~~~~~~~~~~~~~&\nonumber\\
\lambda_{l1}e_{i_1}+...+\lambda_{lk}e_{i_k}+\beta_{l1}e_{x_1}+...+\beta_{lm} e_{x_m}+\mu_{l}e_{j_l}&=v_l\label{coeffs3},
\end{align}
where $v_1,...,v_l$ belong to $\mathcal{B}_{1,*,Y_1}$. Now, consider the following vectors
\begin{align}
\overrightarrow{\lambda_i}=[\lambda_{i,k'+1}~ \lambda_{i,k'+2}~ ...~ \lambda_{ik}]^T,~i\in[l].
\end{align}

There are $l$ of these vectors in $\mathbb{R}^{k-k'}$ and since we assumed that $k'+l>k$, these vectors should be linearly dependent. This implies that there exist not-all-zero coefficients $a_i$, $i\in\{1,...,l\}$ such that $\sum_{i=1}^l a_i \overrightarrow{\lambda_i}=\overrightarrow{0}$. Multiplying each $a_i$ by the corresponding equation in (\ref{coeffs1})-(\ref{coeffs3}) and then adding the resulting equations yields
\begin{align}\label{lincomb}
\sum_{n=1}^{k'}\left(\sum_{i=1}^l a_i \lambda_{in}\right) e_{i_n} + \sum_{n=1}^{m}\left(\sum_{i=1}^l a_i \beta_{in}\right) e_{x_n} + \sum_{i=1}^l (a_i \mu_i) e_{j_i}=\sum_{i=1}^l a_i  v_i,
\end{align}
where the vectors in (\ref{lincomb}) are non-zero since the coefficients $a_i$ are not all zeros and all the coefficients $\mu_i$ are non-zero. The RHS of (\ref{lincomb}) is a vector in $\mathcal{B}_{1,*,Y_1}$, which implies that there exists a linear combination of the bases of $S_{J\cup X^c}$ which lies in $\mathcal{B}_{1,*,Y_1}$, therefore contradicting the fact that $J\in\mathcal{I}_{1,X,Y_1}$. Hence, this property is also true.
\end{itemize}

Having proven the above properties, it is now verified that 
$M_{1,X,Y_1}=(X,\mathcal{I}_{1,X,Y_1})$ is a matroid.

To complete the proof of Claim \ref{matroid}, we need to show that the rank function of $M_{1,X,Y_1}$ is equal to $r_{1,X,Y_1}(J)=|J|-\dim (\mathcal{S}_{J\cup X^c}\cap \mathcal{B}_{1,*,Y_1}) , J\subseteq X$. We do so through the following two steps.

\begin{itemize}
\item For any $I\subseteq J$ such that $\dim (\mathcal{S}_{I\cup X^c} \cap \mathcal{B}_{1,*,Y_1}) =0$, we have that $(\mathcal{S}_{J\cup X^c} \cap \mathcal{S}_{I\cup X^c} )  + (\mathcal{S}_{J\cup X^c}  \cap \mathcal{B}_{1*,Y_1})  \subseteq \mathcal{S}_{J\cup X^c}$.\footnote{For two subspaces $\mathcal{U}$ and $\mathcal{V}$, $\mathcal{U}+\mathcal{V}$ stands for $\text{span}(\mathcal{U}\cup\mathcal{V})$.} Therefore, since $\dim (\mathcal{S}_{I\cup X^c} \cap \mathcal{B}_{1,*,Y_1}) =0$, we will get $\dim (\mathcal{S}_{J\cup X^c} \cap \mathcal{S}_{I\cup X^c} )  + \dim (\mathcal{S}_{J\cup X^c}  \cap \mathcal{B}_{1,*,Y_1})  \leq \dim(\mathcal{S}_{J\cup X^c} ) $. But $\mathcal{S}_{J\cup X^c} \cap \mathcal{S}_{I\cup X^c} =\mathcal{S}_{I\cup X^c} $. Hence, $\dim (\mathcal{S}_{I\cup X^c} )  \leq   \dim(\mathcal{S}_{J\cup X^c} ) - \dim (\mathcal{S}_{J\cup X^c}  \cap \mathcal{B}_{1,*,Y_1}) $, implying that $r_{1,X}(J)\leq |J|-\dim (\mathcal{S}_{J\cup X^c} \cap \mathcal{B}_{1,*,Y_1}) $.

\item

Now, we only need to show that there exists some $K\subseteq J $ for which $|K| =   |J |- \dim (\mathcal{S}_{J\cup X^c}  \cap \mathcal{B}_{1,*,Y_1}) $ and $\dim (\mathcal{S}_{K\cup X^c}  \cap \mathcal{B}_{1,*,Y_1}) =0$. By definition, vectors $\{e_i: i\in J\cup X^c\}$ form a basis for the vector space $\mathcal{S}_{J\cup X^c} $.  Suppose that the vectors in $\{f_1,...,f_d\}$ form a basis for the space $\mathcal{S}_{J\cup X^c} \cap \mathcal{B}_{1,*,Y_1}$. Therefore, since $\dim(\mathcal{S}_{X^c}\cap \mathcal{B}_{1,*,Y_1})=0$, the vectors in $\{f_1,...,f_d\} \cup\{e_i:i\in X^c\}$ form a basis for 
$\mathcal{S}_{J\cup X^c} \cap \mathcal{B}_{1,*,Y_1} + \mathcal{S}_{X^c}$. Then, by the Steinitz exchange lemma, there exists a subset $K\subseteq J$ such that the vectors in $\{f_1,...,f_d\}\cup\{e_i:i\in K\cup X^c\}$ form a basis for the space $\mathcal{S}_{J\cup X^c}$. This implies that $\mathcal{S}_{K}$ has no intersection with $\mathcal{S}_{J\cup X^c} \cap \mathcal{B}_{1,*,Y_1} + \mathcal{S}_{X^c}$ and therefore with $\mathcal{B}_{1,*,Y_1}$ (except for the zero vector) and $|K| =   |J\cup X^c|- \dim (\mathcal{S}_{J\cup X^c} \cap \mathcal{B}_{1,*,Y_1} + \mathcal{S}_{X^c})=|J|-\dim (\mathcal{S}_{J\cup X^c}\cap \mathcal{B}_{1,*,Y_1}) $, suggesting that $K$ is the actual desired subset of $J$.

\end{itemize}

The above two steps imply that $r_{1,X,Y_1}(J)=|J|-\dim (\mathcal{S}_{J\cup X^c}\cap \mathcal{B}_{1,*,Y_1}) , J\subseteq X$. The same arguments hold for any $M_{i,X,Y_i}=(X,\mathcal{I}_{i,X,Y_i}), i\in[K]$ as well. This completes the proof.

\section{Proof of Lemma~\ref{lem:non_overlapping}}\label{app:non_overlapping}
Fix a design for the transmit beamforming matrices $B_i: i \in [K]$ such that~\eqref{eq:alignment} holds. For each $r \in [K]: |I_r| = 2$, fix a subset $J_r \subseteq [n]$ such that the condition in~\eqref{eq:alignment} is satisfied for the chosen subsets; it is easy to verify that $|J_r| \geq \tau$ for any selected subset. We then set $J_r = \phi$ for every $r \in [K]$ such that $|I_r| < 2$. 

For each $r \in [K]$ such that $I_r=\{r_1,r_2\}$, we choose the new beamforming matrices $B_{r1}^{(\text{new})}$ and $B_{r2}^{(\text{new})}$ as follows. let $J'_r=J_r \backslash \left(J_{r1} \cup J_{r2}\right)$ and $\mathcal{S}'_r=\left(\mathcal{S}_{J_r} \backslash \mathcal{S}_{J_{r1}} \cup \mathcal{S}_{J_{r2}}\right)$; since \eqref{eq:alignment} is satisfied, we know that the following holds,
\begin{equation}
\dim(\mathcal{S}'_r \cap \mathcal{B}_{r1})+\dim(\mathcal{S}'_r \cap \mathcal{B}_{r2}) \geq |J'_r|+\tau,
\end{equation}
and hence, we also know that $|J'_r| \geq \tau$. For $i \in \{1,2\}$, let $B'_{ri}$ be an $n \times nd$ matrix with $nd$ columns forming a basis for $B_{ri}$ and a subset of these columns form a basis for $\mathcal{S}'_r \cap \mathcal{B}_{ri}$; we then fix $J_r^{(\text{new})} \subseteq J'_r: |J_r^{(\text{new})}|=\tau$ and replace the columns that form a basis for $\mathcal{S}'_r \cap \mathcal{B}_{ri}$ with an equal number of linearly independent columns in $\mathcal{S}'_r$ and $\tau$ of the new columns form a basis for $\mathcal{S}_{J_r^{(\text{new})}}$ to construct the matrix $B_{ri}^{(\text{new})}$ from $B'_{ri}$.

After performing the above step, it is straightforward to verify that the following holds,
\begin{gather}
\forall r \in [K]: I_r=\{r_1,r_2\}, \nonumber\\\mathcal{S}_{J_r^{(\text{new})}} \subseteq \mathcal{B}_{r1}^{(\text{new})} \cap \mathcal{B}_{r2}^{(\text{new})}, \dim\left(\mathcal{S}_{J_r^{(\text{new})}} \cap \mathcal{B}_r^{(\text{new})}\right)=0.\label{eq:alignmentsatisfied}
\end{gather}
and hence, the new beamforming matrices satisfy the condition in~\eqref{eq:newalignment}.

\bibliographystyle{IEEEtran}
{\footnotesize
\bibliography{navid}}

\begin{thebibliography}{10}
\providecommand{\url}[1]{#1}
\csname url@samestyle\endcsname
\providecommand{\newblock}{\relax}
\providecommand{\bibinfo}[2]{#2}
\providecommand{\BIBentrySTDinterwordspacing}{\spaceskip=0pt\relax}
\providecommand{\BIBentryALTinterwordstretchfactor}{4}
\providecommand{\BIBentryALTinterwordspacing}{\spaceskip=\fontdimen2\font plus
\BIBentryALTinterwordstretchfactor\fontdimen3\font minus
  \fontdimen4\font\relax}
\providecommand{\BIBforeignlanguage}[2]{{%
\expandafter\ifx\csname l@#1\endcsname\relax
\typeout{** WARNING: IEEEtran.bst: No hyphenation pattern has been}%
\typeout{** loaded for the language `#1'. Using the pattern for}%
\typeout{** the default language instead.}%
\else
\language=\csname l@#1\endcsname
\fi
#2}}
\providecommand{\BIBdecl}{\relax}
\BIBdecl

\bibitem{path-matching}
W.~H. Cunningham and J.~F. Geelen, ``The optimal path-matching problem,''
  \emph{Combinatorica}, vol.~17, no.~3, pp. 315--337, 1997.

\bibitem{lovasz}
L.~Lov{\'a}sz, ``On determinants, matchings, and random algorithms.'' in
  \emph{FCT}, vol.~79, 1979, pp. 565--574.

\bibitem{factor}
W.~T. Tutte, ``The factorization of linear graphs,'' \emph{Journal of the
  London Mathematical Society}, vol.~1, no.~2, pp. 107--111, 1947.

\bibitem{sch}
J.~T. Schwartz, ``Fast probabilistic algorithms for verification of polynomial
  identities,'' \emph{Journal of the ACM (JACM)}, vol.~27, no.~4, pp. 701--717,
  1980.

\bibitem{zip}
R.~Zippel, \emph{Probabilistic algorithms for sparse polynomials}.\hskip 1em
  plus 0.5em minus 0.4em\relax Springer, 1979.

\bibitem{localview}
V.~Aggarwal, A.~S. Avestimehr, and A.~Sabharwal, ``On achieving local view
  capacity via maximal independent graph scheduling,'' \emph{Information
  Theory, IEEE Transactions on}, vol.~57, no.~5, pp. 2711--2729, 2011.

\bibitem{jafar}
S.~A. Jafar, ``Topological interference management through index coding,''
  \emph{Information Theory, IEEE Transactions on}, vol.~60, no.~1, pp.
  529--568, 2014.

\bibitem{birk}
Y.~Birk and T.~Kol, ``Informed-source coding-on-demand (iscod) over broadcast
  channels,'' in \emph{INFOCOM '98. Seventeenth Annual Joint Conference of the
  IEEE Computer and Communications Societies. Proceedings. IEEE}, vol.~3, Mar
  1998, pp. 1257--1264.

\bibitem{baryossef}
Z.~Bar-Yossef, Y.~Birk, T.~Jayram, and T.~Kol, ``Index coding with side
  information,'' in \emph{Foundations of Computer Science, 2006. FOCS '06. 47th
  Annual IEEE Symposium on}, Oct 2006, pp. 197--206.

\bibitem{rouayheb}
S.~El~Rouayheb, A.~Sprintson, and C.~Georghiades, ``On the index coding problem
  and its relation to network coding and matroid theory,'' \emph{Information
  Theory, IEEE Transactions on}, vol.~56, no.~7, pp. 3187--3195, 2010.

\bibitem{local_coloring_index_coding}
K.~Shanmugam, A.~Dimakis, and M.~Langberg, ``Local graph coloring and index
  coding,'' in \emph{Information Theory Proceedings (ISIT), 2013 IEEE
  International Symposium on}, July 2013, pp. 1152--1156.

\bibitem{topology}
N.~Naderializadeh and A.~S. Avestimehr, ``Interference networks with no csit:
  Impact of topology,'' \emph{Information Theory, IEEE Transactions on},
  vol.~61, no.~2, pp. 917--938, Feb 2015.

\bibitem{topology_isit}
------, ``Impact of topology on interference networks with no csit,'' in
  \emph{Information Theory Proceedings (ISIT), 2013 IEEE International
  Symposium on}.\hskip 1em plus 0.5em minus 0.4em\relax IEEE, 2013, pp.
  394--398.

\bibitem{multialign}
T.~Gou, C.~R. da~Silva, J.~Lee, and I.~Kang, ``Partially connected interference
  networks with no csit: Symmetric degrees of freedom and multicast across
  alignment blocks,'' \emph{Communications Letters, IEEE}, vol.~17, no.~10, pp.
  1893--1896, 2013.

\bibitem{schrijver}
A.~Schrijver, \emph{Combinatorial optimization: polyhedra and
  efficiency}.\hskip 1em plus 0.5em minus 0.4em\relax Springer Science \&
  Business Media, 2003, vol.~24.

\bibitem{avoidance_jafar}
X.~Yi, H.~Sun, S.~A. Jafar, and D.~Gesbert, ``Fractional coloring (orthogonal
  access) achieves all-unicast capacity (dof) region of index coding (tim) if
  and only if network topology is chordal,'' \emph{arXiv preprint
  arXiv:1501.07870}, 2015.

\bibitem{ICC}
N.~Naderializadeh, A.~E. Gamal, and A.~S. Avestimehr, ``Topological
  interference management with just retransmission: What are the" best"
  topologies?'' \emph{arXiv preprint arXiv:1502.03147}, 2015.

\end{thebibliography}

\end{document}